\documentclass{article}
\usepackage[T1]{fontenc}

\usepackage{geometry}
\usepackage[many]{tcolorbox}    

\usepackage{amsmath}
\usepackage{amsthm}
\usepackage{amsfonts}
\usepackage{amssymb}

\usepackage{bbm}

\usepackage{cite}

\usepackage{subcaption}

\usepackage{thmtools,thm-restate}

\usepackage{tikz}
\usetikzlibrary{positioning}
\usetikzlibrary{calc}

\usepackage{pgfplots}
\pgfplotsset{width=8cm,compat=1.9}

\usepackage{url}

\let\emptyset\varnothing

\newtheorem*{theorem*}{Theorem}

\newtheorem{definition}{Definition}[section]
\newtheorem{proposition}[definition]{Proposition}
\newtheorem{lemma}[definition]{Lemma}
\newtheorem{theorem}[definition]{Theorem}
\newtheorem{corollary}[definition]{Corollary}

\newtheorem{remark}[definition]{Remark}
\newtheorem{krule}[definition]{Rule}
\newtheorem{problem}[definition]{Problem}

\newcommand{\A}{\mathcal{A}}
\newcommand{\G}{\mathcal{G}}

\newcommand{\C}{\mathcal{C}}
\newcommand{\Q}{\mathcal{Q}}
\newcommand{\F}{\mathcal{F}}
\newcommand{\R}{\mathcal{R}}
\newcommand{\D}{\mathcal{D}}

\title{Computing Clique Cover with Structural Parameterization}
\author{Ahammed Ullah\thanks{Purdue University}}

\date{}

\begin{document}

\maketitle

An abundance of real-world problems manifest as covering edges and/or vertices of a graph with cliques that are optimized for some objectives. We consider different structural parameters of graph, and design fixed-parameter tractable algorithms for a number of clique cover problems. Using a set representation of graph, we introduce a framework for computing clique cover with different objectives. We demonstrate use of the framework for a variety of clique cover problems. Our results include a number of new algorithms with exponential to double exponential improvements in the running time.

\section{Introduction}
\label{sec_intro}

A set of cliques is an edge (resp. vertex) clique cover of a graph\footnote{Throughout the paper we assume graphs are finite, undirected, and simple, i.e., contains no self-loop and parallel edges.} if every edge (resp. vertex) is contained in at least one of the cliques in the set. In a clique cover problem, we seek to compute a clique cover that is optimized for some objectives such as minimizing number of cliques in the cover. The clique cover problems we are pursuing in this study are \emph{NP-hard}, and inapproximable within any constant factor unless $P=NP$. For applications of clique cover problems, see \cite{agarwal1994can, blanchette2012clique, cooley2021parameterized, ennis2012assignment, feldmann2020fixed, gramm2007algorithms, helling2018constructing, kou1978covering, markham2020measurement, piepho2004algorithm, rajagopalan2000handling, roberts1985applications, strash2022effective}.

One of the ways of tackling intractability of \emph{NP-hard} problems is to study the problems through the lens of parameterized complexity \cite{cygan2015parameterized, downey2013fundamentals, flum2006parameterized}. In parameterized complexity\footnote{See Section \ref{prelim_complexity} for a brief introduction of relevant terms.}, a problem is called fixed-parameter tractable (FPT) with respect to a parameter $p$ if problem instance of size $n$ can be solved in $f(p)n^{O(1)}$ time, where $f(p)$ is computable and independent of $n$, but can grow arbitrarily with $p$. In terms of applicability, the hope is that for small values of parameter an FPT problem would be solvable in a reasonable amount of time.

\subsection{Motivations}
\label{intro_motivation}

A number of clique cover problems have been shown to be FPT with respect to the number of cliques in a solution \cite{cooley2021parameterized, feldmann2020fixed, gramm2009data}. These FPT running times have super exponential dependence on the number of cliques, making them unlikely to be useful for graphs which are not sufficiently dense. For example, deciding whether edges of a graph can be covered with at most $k$ cliques is solvable in $2^{2^{O(k)}}n^{O(1)}$ time \cite{gramm2009data}. The double exponential dependence on $k$ is necessary, assuming the exponential time hypothesis \cite{cygan2016known}. We ask to what extent super exponential dependence on $k$ arises, and whether structural parameters of graph can be used to tame dependence on $k$.

Known frameworks of computing clique covers are based on enumeration of binary matrices of dimension $O(k) \times O(k)$ \cite{cooley2021parameterized, feldmann2020fixed} ($k$ is the number of cliques in a solution), or enumeration of maximal cliques of different subgraphs \cite{ennis2012assignment, gramm2009data}. Clearly, enumeration of binary matrices of dimension $O(k) \times O(k)$ disregards any sparsity measures of graph. Also, time and space complexities of algorithm designed using maximal clique enumeration may have prohibitively large dependence on input size. For example, \cite{ennis2012assignment} have described an algorithm for the problem of covering the edges of a graph with cliques such that the number of individual assignments of vertices to the cliques is minimum. For a graph with $n$ vertices, the algorithm of \cite{ennis2012assignment} takes $2^{2^{O(n)}}$ time and $2^{O(n)}$ space. We investigate amenability of different frameworks for designing better algorithms that also take sparsity of graph into account.

Many clique cover problems are unlikely to be FPT with respect to certain choices of parameters. For example, covering vertices of a graph with at most $k$ cliques is not FPT with respect to $k$ unless $P=NP$. We investigate whether structural parameters of graph would be helpful for designing FPT algorithms for these problems (and corresponding graph coloring problems).

\subsection{Problems}
\label{intro_problems}

We describe the clique cover problems we have studied in this paper.

Covering the edges of a graph with at most $k$ cliques is an \emph{NP-complete} problem \cite{kou1978covering, orlin1977contentment}. The problem has many applications \cite{roberts1985applications}, and appears in different guises such as keyword conflict \cite{kou1978covering}, intersection graph basis \cite{garey1979computers}, indeterminate string \cite{helling2018constructing}. For ease of reference, the problem is outlined as follows.

\begin{center}
\fbox{
\parbox{6.0in}{
\emph{Edge Clique Cover} (\emph{ECC})

\textbf{Input:} A graph $G=(V,E)$, and a nonnegative integer $k$.

\textbf{Output:} If one exists, a set of at most $k$ cliques of $G$ such that every edge of $G$ is contained in at least one of the cliques in the set; otherwise \emph{NO}.
}}    
\end{center}

Motivated by applications from applied statistics \cite{gramm2007algorithms, piepho2004algorithm}, \cite{ennis2012assignment} have studied the \emph{assignment-minimum clique cover} problem that seeks to cover edges of a graph with cliques such that the number of individual assignments of vertices to cliques is minimum. \cite{ennis2012assignment} have demonstrated that minimizing number of cliques and number of individual assignments of vertices simultaneously is not always possible. Therefore, number of cliques is not an appropriate choice for \emph{natural parameterization} of the problem. We introduce the following parameterized problem for \emph{assignment-minimum clique cover}.

\begin{center}
\fbox{\parbox{6.0in}{
\emph{Assignment Clique Cover} (\emph{ACC})

\textbf{Input:} A graph $G=(V,E)$, and a nonnegative integer $t$.

\textbf{Output:} If one exists, a clique cover $\C$ of $G$ such that (1) every edge of $G$ is contained in at least one of the cliques in $\C$, and (2) $\sum_{C_l \in \C} |C_l| \leq t$; otherwise \emph{NO}.
 }}    
\end{center}

\cite{feldmann2020fixed} have introduced the following parameterized clique cover problem.

\begin{center}
\fbox{\parbox{6.0in}{
\emph{Weighted Edge Clique Partition} (\emph{WECP})

\textbf{Input:} A graph $G=(V,E)$, a weight function $w: E \rightarrow \mathbb{Z}_{>0}$, and a nonnegative integer $k$.

\textbf{Output:} If one exists, a clique cover $\C$ of $G$ such that (1) $|\C| \leq k$, and (2) each edge $e \in E$ appears in exactly $w(e)$ cliques of $\C$; otherwise \emph{NO}.
 }}    
\end{center}

Motivated by applications from identification of gene co-expression modules, \cite{cooley2021parameterized} have introduced a generalization of \emph{WECP} problem as follows.

\begin{center}
\fbox{\parbox{6.0in}{
\emph{Exact Weighted Clique Decomposition} (\emph{EWCD})

\textbf{Input:} A graph $G=(V,E)$, a weight function $w: E \rightarrow \mathbb{R}_{>0}$, and a nonnegative integer $k$.

\textbf{Output:} If one exists, a clique cover $\C$ of $G$ and $\gamma_i \in \mathbb{R}_{>0}$ for all $C_i \in \C$ such that (1) $|\C|\leq k$, (2) and for each edge $e= \{u,v\} \in E$, $\sum_{\{u,v\} \in C_i} \gamma_i = w(e)$; otherwise \emph{NO}.
 }}    
\end{center}

The \emph{vertex clique cover} problem (abbreviated \emph{VCC}) is the problem of deciding whether vertices of a graph can be covered (or partitioned) using at most $k$ cliques. \emph{VCC} is \emph{NP-complete} \cite{karp1972reducibility}. In parameterized complexity, many variants of the \emph{NP-complete} problems from \cite{karp1972reducibility} had been extensively studied such as \emph{vertex cover} \cite{guo2007parameterized}, \emph{graph coloring} \cite{fiala2011parameterized, jansen2019computing}. \emph{VCC} in this regard had been totally unexplored. We introduce following generalization of \emph{VCC}.

\begin{center}
\fbox{\parbox{6.0in}{
\emph{Link Respected Clique Cover} (\emph{LRCC})

\textbf{Input:} A graph $G=(V,E)$, a nonnegative integer $k$, and a set of edges $E^* \subseteq E$.

\textbf{Output:} If one exists, a set of cliques $\C$ of $G$ such that (1) $|\C| \leq k$, (2) every vertex of $G$ is contained in at least one of the cliques in $\C$, (3) and every edge $\{x,y\} \in E^*$ is contained in at least one of the cliques in $\C$; otherwise \emph{NO}.
 }}    
\end{center}

\emph{LRCC} can be seen as the problem of covering vertices of a network with $k$ communities such that certain set of links (edges) are preserved in the cover. Note that if $E^*=E$, then \emph{LRCC} is equivalent to \emph{ECC}. On the other hand, if $E^* = \emptyset$, then \emph{LRCC} reduces to \emph{VCC}.

\emph{Colorability} is the problem of coloring vertices of a graph such that each pair of adjacent vertices receive difference colors. \emph{VCC} is equivalent to \emph{Colorability} in the complement graph. The following problem is equivalent to \emph{LRCC} in the complement graph.

\begin{center}
\fbox{\parbox{6.0in}{
\emph{Pairwise Mutual Colorability} (\emph{PMC})

\textbf{Input:} A graph $G=(V,E)$, a nonnegative integer $k$, and a set of pairs of non-adjacent vertices $\F$.

\textbf{Output:} If exist, assignments of at most $k$ colors to the vertices of $G$ such that (1) each vertex receive at least one color, (2) no pair of adjacent vertices share any color, and (3) each pair of vertices $\{x,y\} \in \F$ share at least one color; otherwise \emph{NO}.
 }}    
\end{center}

\emph{PMC} captures formulation of a scheduling problem where each pair of adjacent vertices are in conflict, and we want a conflict free assignments of vertices into $k$ slots such that each pair of vertices in $\F$ has at least one common slot. If $\F = \emptyset$, then \emph{PMC} reduces to \emph{Colorability}.

From the structural parameters\footnote{See Section \ref{prelim_parameters} for a brief description of different parameters and their relationships.} that capture graph sparsity, we focus on \emph{degeneracy} and \emph{clique number} of graph. Many structural parameters (such as \emph{arboricity}, \emph{thickness}, \emph{treewidth}, \emph{vertex cover number}) are either within constant factors of or upper bounds on our chosen structural parameters. Therefore, our algorithms are also FPT with respect to many structural parameters that we do not explicitly focus on.

\subsection{Our Results}
\label{intro_contribution}

For a graph $G=(V,E)$, let $n=|V|$ denote number of vertices, $m=|E|$ denote number of edges, $d$ denote the \emph{degeneracy}, $\beta$ denote the \emph{clique number}, and $\alpha$ denote the \emph{independence number}.

Our results are obtained from two different algorithmic frameworks: one described in Section \ref{sec_ecrs} (denoted \emph{Framework-1}), and the other described in Section \ref{sec_bblock} and Section \ref{sec_algos} (denoted \emph{Framework-2}).

Using \emph{Framework-1}, we obtain the following result for \emph{ECC}.

\begin{theorem}
\label{thm_ecc}
\emph{ECC} parameterized by $d$ and $k$  has an FPT algorithm running in $1.4423^{dk}n^{O(1)}$ time.
\end{theorem}

With kernelization, we achieve a factor of $2^{\Omega(\alpha^2)}$ times faster algorithm than the algorithm of \cite{gramm2009data}. For $\alpha = \Omega(n)$, the improvement in the running time amounts to a factor of $2^{2^{O(k)}}$. We point out that the factor of $2^{\Omega(\alpha^2)}$ improvement in the running time is irrespective of graph density, i.e., the algorithm of Theorem \ref{thm_ecc} is faster than the algorithm of \cite{gramm2009data} regardless of whether the \emph{degeneracy} of input graph is bounded or not.

Using \emph{Framework-2}, we obtain the following result for \emph{ECC}.

\begin{theorem}
\label{thm_ecc2}
\emph{ECC} parameterized by $\beta$ and $k$ has an FPT algorithm running in $2^{\beta k\log k}n^{O(1)}$ time.\footnote{Unless otherwise specified, base of logarithms in this paper is two.}
\end{theorem}

For $\beta = o(d / \log d)$, the algorithm of Theorem \ref{thm_ecc2} is asymptotically faster than the algorithm of Theorem \ref{thm_ecc} by a factor $2^{O(dk)}$. For many instances of \emph{ECC}, $\beta = o(d/\log d)$ is a very mild requirement: $d$ can grow linearly with the input size, while $\beta$ remains constant.

\cite{ennis2012assignment} have claimed that the \emph{assignment-minimum clique cover} is NP-hard; but we have not found any proof of this claim in the literature. We fill this gap by proving the following claim.

\begin{theorem}
\label{thm_acc}
\emph{ACC} is \emph{NP-complete}.
\end{theorem}

We describe two FPT algorithms for \emph{ACC}, one is obtained using \emph{Framework-1}, and the other is obtained using \emph{Framework-2}. The following is our best FPT running time for \emph{ACC} that we obtain using \emph{Framework-2}.

\begin{theorem}
\label{thm_acc2}
\emph{ACC} has an FPT algorithm running in $4^{t \log t}n^{O(1)}$ time.
\end{theorem}

In contrast to $2^{2^{O(n)}}$ running time and $2^{O(n)}$ space complexity of the algorithm described by \cite{ennis2012assignment}, we obtain the following result from Theorem \ref{thm_acc2}.

\begin{corollary}
\label{cor_acc}
An \emph{assignment-minimum clique cover} of a graph $G$ can be computed in $2^{O(m \log n)}$ time, using $O(mn)$ space.
\end{corollary}

Rest of our results are obtained using \emph{Framework-2}. For \emph{WECP}, we have the following result.

\begin{theorem}
\label{thm_wecp}
\emph{WECP} parameterized by $\beta$ and $k$ has an FPT algorithm running in $2^{\beta k\log k}n^{O(1)}$ time.
\end{theorem}

For \emph{WECP}, \cite{feldmann2020fixed} have described an FPT algorithm with $2^{O(k^{3/2}w^{1/2}\log (k/w))}n^{O(1)}$ running time, where $w$ is the maximum edge weight. For $\beta = o((\sqrt{kw}\log (k/w)) / \log k)$, the algorithm of Theorem \ref{thm_wecp} improves the running time by a factor of $2^{O(k^{3/2}w^{1/2}\log (k/w))}$. This is a significant improvement, considering the fact that $\beta$ is bounded by the graph size, while $k$ and $w$ could be arbitrarily large. For nontrivial instances of \emph{WECP}, $k$ could be up to (but not including) $mw$. On the other hand, if the maximum edge weight $w$ is bounded by some constant, then for $\beta = o(\sqrt{k})$, the algorithm of Theorem \ref{thm_wecp} is $2^{O(k^{3/2} \log k)}$ times faster than the algorithm of \cite{feldmann2020fixed}.

For \emph{EWCD}, we have the following result.

\begin{theorem}
\label{thm_ewcd}
\emph{EWCD} parameterized by $\beta$ and $k$ has an FPT algorithm running in $2^{O(\beta k\log k)}n^{O(1)}L$ time, where $L$ is the number of bits required for input representation\footnote{More precisely, $L$ is the number of bits required to encode the input in a linear program (see (\ref{aewcd_lp})).}.
\end{theorem}

For \emph{EWCD}, \cite{cooley2021parameterized} have described an FPT algorithm with running time $O(4^{k^2}k^2(32^k+k^3L))$, where $L$ is the number of bits required for input representation. For $\beta = o(k / \log k)$, the algorithm of Theorem \ref{thm_ewcd} improves the running time by a factor of $2^{O(k^2)}$. For the cases when all weights are restricted to integers, \cite{cooley2021parameterized} have described an algorithm with running time $O(4^{k^2}32^kw^kk)$, where $w$ is the maximum edge weight. For these cases, if $\beta = o(k / \log k)$, then an adaptation of the algorithm of Theorem \ref{thm_ewcd} also improves the running time by a factor of $2^{O(k^2)}$.

For \emph{LRCC}, we have the following result.

\begin{theorem}
\label{thm_lrcc}
\emph{LRCC} parameterized by $\beta$ and $k$ has an FPT algorithm running in $2^{\beta k\log k}n^{O(1)}$ time.
\end{theorem}

Through a \emph{parameterized reduction}, we obtain the following result from Theorem \ref{thm_lrcc}.

\begin{corollary}
\label{cor_pmc}
\emph{PMC} parameterized by $\alpha$ and $k$ has an FPT algorithm running in $2^{\alpha k\log k}n^{O(1)}$ time.
\end{corollary}

\subsection{Overview of Our Techniques}
\label{intro_overview}

Our algorithms are based on combination of data reduction rules and bounded search tree algorithm. We use two different frameworks for designing bounded search tree algorithms. We make sure that all of our search tree algorithms are compatible with the existing data reduction rules (and conceivably with many rules yet to be found). For an instance of \emph{WECP} (or \emph{EWCD}), \emph{kernelization} produces an instance of a more general problem. For these cases, to make use of the \emph{kernelization}, we describe search tree algorithms for the general problems.

Our first framework (Section \ref{sec_ecrs}) for search tree algorithms is based on enumeration of cliques of subgraphs, introduced by \cite{gramm2009data}. With an edge selection rule, \cite{gramm2009data} have attempted to bound the branching factors of a search tree. But the rule itself does not lead to any provable bound on the branching factors. As a result, the branching factors of a search tree of \cite{gramm2009data} is only bounded by a function of $k$ that is double exponential in $k$ (obtained through kernelization). We enumerate cliques of subgraphs such that the size of each subgraph is bounded by a polynomial function of \emph{degeneracy}. This allows us to bound the branching factors of a search tree with a single exponential function of \emph{degeneracy} (or other parameters).

We introduce a general way of computing clique cover based on two concepts that we call \emph{locally minimal clique cover} and \emph{implicit set representation}. The concepts may be of independent interest, and we devote a fair amount of exposition highlighting their characterizations and relationships (Section \ref{sec_bblock}). We use the concepts to demonstrate a general framework for designing search tree algorithms for clique cover problems (Section \ref{sec_algos}).

The concept of \emph{locally minimal clique cover} is based on a relaxation of a global minimality of clique cover. With this relaxation, systematic exploration of clique covers of a graph becomes straightforward. Every graph corresponds to some family of sets called \emph{set representation} of the graph \cite{mckee1999topics}. The union of the sets in a \emph{set representation} of a graph corresponds to an \emph{edge clique cover} of the graph; but this correspondence does not immediately lead to efficient construction of \emph{edge clique cover}. \emph{Implicit set representation} of a graph is a relaxation of \emph{set representation} of the graph that helps efficient construction of \emph{edge clique cover} of a graph. The two concepts, \emph{locally minimal clique cover} and \emph{implicit set representation}, nicely fit together, and give us a general way of computing clique cover with different objectives.

\subsection{Related Work}
\label{intro_related}

\cite{gramm2009data} have described the first FPT algorithm for \emph{ECC}. The algorithm of \cite{gramm2009data} is based on a kernel with $2^k$ vertices and a search tree algorithm that takes $2^{2^{O(k)}}$ time. Unless the polynomial hierarchy collapses, no polynomial kernel exists for \emph{ECC} \cite{cygan2014clique}. The FPT running time of \cite{gramm2009data} is essentially optimal with respect to $k$ \cite{cygan2016known}. \cite{blanchette2012clique} have described an FPT algorithm for \emph{ECC} parameterized by \emph{treewidth}, and for planar graphs an FPT algorithm for \emph{ECC} parameterized by \emph{branchwidth}.

The algorithm of \cite{ennis2012assignment} is based on the observation that an \emph{edge clique cover} of a graph consisting of all maximal cliques of the graph must contain an \emph{assignment-minimum clique cover} of the graph. The algorithm starts with all maximal cliques of a graph, and enumerates all possible individual assignments of vertices that can be removed, maintaining an \emph{edge clique cover}. The requirement of having (simultaneously) all maximal cliques of a graph with $n$ vertices results in $O(n3^{n/3}) = 2^{O(n)}$ space requirement, and enumeration of the choices of removing individual assignments of vertices results in $2^{2^{O(n)}}$ running time.

\cite{feldmann2020fixed} have described an FPT algorithm with $2^{O(k^{3/2}w^{1/2}\log (k/w))} n^{O(1)}$ running time, where $w$ is the maximum edge weight. The FPT algorithm of \cite{feldmann2020fixed} is based on a \emph{bi-kernel} with $4^k$ vertices, and a search tree algorithm that enumerates matrices of $\{0,1\}^{ k \times k}$ such that the dot product of any pair of rows in a matrix is bounded by $w$.

Using the \emph{bi-kernel} of \cite{feldmann2020fixed} and a search tree algorithm that enumerates matrices of $\{0,1\}^{2k \times k}$, \cite{cooley2021parameterized} have described two FPT algorithms for \emph{EWCD}. For arbitrary positive weights, \cite{cooley2021parameterized} have described a linear program based search tree algorithm with running time $O(4^{k^2}k^2(32^k+k^3L))$, where $L$ is the number of bits required for input representation. For the cases when all weights are restricted to integers, \cite{cooley2021parameterized} have described an integer partitioning based algorithm with running time $O(4^{k^2}32^kw^kk)$, where $w$ is the maximum edge weight.

In the literature, we have not found any study on parameterizations of \emph{VCC}. On the other hand, \emph{Colorability} parameterized by \emph{treewidth} is known to be FPT \cite{cygan2015parameterized, flum2006parameterized}. Preceding result is not useful for obtaining FPT algorithm for \emph{VCC}, since no \emph{parameterized reduction} is known from any parameterization of \emph{VCC} to \emph{Colorability} parameterized by \emph{treewidth}. Similarly, no \emph{parameterized reduction} is known from any parameterization of \emph{VCC} to other parameterizations of \emph{Colorability} that are known to be FPT \cite{jansen2019computing}.

\subsection{Paper Organization}

In Section \ref{sec_prelim}, we provide a brief description of relevant structural parameters and parameterized complexity terms. In Section \ref{sec_ecrs}, we describe FPT algorithms for \emph{ECC}, \emph{ACC}, and a proof of \emph{NP-completeness} of \emph{ACC}. The building blocks of our new framework are described in Section \ref{sec_bblock}. In Section \ref{sec_algos}, using our new framework, we describe a new set of algorithms for the problems listed in Section \ref{intro_problems}. We conclude with a discussion on implementations and open problems in Section \ref{sec_open}.

\section{Preliminaries}
\label{sec_prelim}

In this section, we describe relevant structural parameters and their relationships. We also provide a brief description of the relevant terms from parameterized complexity.

\subsection{Notations}
\label{prelim_notations}

We describe commonly used notations used in the paper. Additional notations are introduced throughout the paper where they are appropriate for introduction.

For a vertex $x \in V$, let $N(x) = \{y | \{x,y\} \in E\}$ denote the neighbourhood of $x$ in $G$, and let $N[x] = N(x) \cup \{x\}$ denote the closed neighbourhood of $x$ in $G$. A vertex $x$ is called an isolated vertex if $N(x) = \emptyset$. Let $\Delta$ denote the maximum degree of $G$, i.e., $\Delta = \max_{x \in V} |N(x)|$. If the end points of an edge $e \in E$ are $x$ and $y$, then we denote the edge $e$ by $\{x,y\}$: it would be clear from the context whether $\{x,y\}$ denote an edge or an arbitrary pair of vertices.

We call $H=(V_H,E_H)$ a subgraph of $G$ if $V_H \subseteq V$ and $E_H$ is a subset of edges of $E$ that have both end points in $V_H$. We call $H$ an induced subgraph of $G$ if $E_H$ is the set all edges of $E$ that have both end points in $V_H$. If $H$ is an induced subgraph of $G$, then we say $V_H$ induces the subgraph $H$ in $G$ and denote the subgraph by $G[V_H]$.

\subsection{Structural Parameters}
\label{prelim_parameters}

\begin{definition}[Degeneracy \cite{lick1970k}]
\label{def1}
The \emph{degeneracy} of a graph $G$ is the smallest value $d$ such that every nonempty subgraph of $G$ contains a vertex that has at most $d$ adjacent vertices in the subgraph.
\end{definition}

Following Definition \ref{def1}, we can find a vertex with at most $d$ neighbours and delete the vertex and its incident edges. Repeating this process on the resulting subgraphs would lead to the empty graph. Therefore, in linear time \cite{matula1983smallest}, we can compute a permutation of vertices of $V$ defined as follows.

\begin{definition}[Degeneracy Ordering]
\label{def2}
The \emph{degeneracy ordering} of a graph $G=(V,E)$ is an ordering of vertices in $V$ such that each vertex has at most $d$ neighbours that come later in the ordering.
\end{definition}

The following two propositions ensue from the preceding definitions.

\begin{proposition}
\label{prop1}
$d+1 \leq n$
\end{proposition}

\begin{proposition}
\label{prop2}
$m < nd$
\end{proposition}

From Proposition \ref{prop2} and the fact that $\sum_{x \in V} |N(x)| = 2m$, we get the following result.

\begin{corollary}
Average degree of $G$ is at most $2d$.    
\end{corollary}

\emph{Degeneracy} is within a constant factors of many sparsity measures of graph such as \emph{arboricity} \cite{chiba1985arboricity} and \emph{thickness} \cite{dean1991thickness}. \emph{Degeneracy} is also known as the \emph{d-core number} \cite{seidman1983network}, \emph{width} \cite{freuder1982sufficient}, \emph{linkage} \cite{kirousis1996linkage}, and is equivalent to the \emph{coloring number} \cite{erdHos1966chromatic}.

\begin{definition}[Vertex cover number]
A set of vertices $S \subset V$ is called a \emph{vertex cover} of $G$ if for every edge $\{x,y\} \in E$, $S$ contains at least one end point of $\{x, y\}$. The \emph{vertex cover number} $\tau$ is the size of a smallest \emph{vertex cover} of $G$.    
\end{definition}

\begin{definition}[Independence number]
A set of vertices $S \subset V$ is called an \emph{independent set} of $G$ if for every pair of distinct vertices $\{x,y\} \in S$, $\{x, y\} \not\in E$. The \emph{independence number} $\alpha$ is the size of a largest \emph{independent set} of $G$.    
\end{definition}

\begin{definition}[Treewidth \cite{robertson1984graph}]
A tree decomposition of a graph $G=(V,E)$ is a pair $(T, \{X_t\}_{t \in V_T})$ where $T$ is a tree whose node set is $V_T$ and each $X_t \subseteq V$ such that following conditions hold: (1) $\cup_{t \in V_T} X_t = V$, (2) for every edge $\{x,y\} \in E$ there exist $t \in V_T$ such that $\{x,y\} \in X_t$, (3) for each $x \in V$ the set of nodes $\{t \in V_T | x \in X_t\}$ induces a connected subtree of $T$.

The width of a tree decomposition $(T, \{X_t\}_{t \in V_T})$ is $\max_{t \in V_T} |X_t| -1$. The treewidth of a graph $G$ is the smallest width of a tree decomposition over all tree decompositions of $G$.
\end{definition}

A graph $G$ with \emph{treewidth} $k$ is also a subgraph of a \emph{$k$-tree} (see Definition 10.2.1 of \cite{downey2013fundamentals}). \emph{$k$-trees} are the chordal graphs all of whose maximal cliques are the same size $k+1$. Since any chordal graph admits an ordering of vertices such that a vertex $x$ and neighbours of $x$ that come later in the ordering form a clique (also known as the \emph{perfect elimination ordering}), \emph{degeneracy} of $G$ is at most the \emph{treewidth} of the graph. Also, \emph{treewidth} of a graph is at most the \emph{vertex cover number} of the graph \cite{fiala2011parameterized}. Therefore, $d \leq \tau$. From the fact that if $S$ is a \emph{vertex cover} of a graph, then $V \backslash S$ is an \emph{independent set} of the graph, we have $\tau + \alpha = n$. Thus we have the following.

\begin{proposition}
\label{prop_dalpha}
$d+\alpha \leq n$
\end{proposition}

\begin{remark}
Degeneracy is a better measure of graph sparsity than many other parameters. To see this, consider the bipartite graph $K_{1,n-1}$ whose \emph{degeneracy} is $1$, but maximum degree is $n-1$. A popular measure of graph sparsity is \emph{treewidth}. A graph may have bounded \emph{degeneracy}, but arbitrarily large \emph{treewidth}: from any graph $G$ with \emph{treewidth} $tw$, we can subdivide the edges (replace edge $\{x,y\}$ of $G$ with two edges $\{x,z\}$ and $\{z,y\}$) and obtain a graph whose \emph{degeneracy} is at most two, but whose \emph{treewidth} remains $tw$. Even extremely sparse graphs can have large \emph{treewidth}: the \emph{degeneracy} of an $n \times n$ gird graph is at most three, whereas its \emph{treewidth} is $\sqrt{n}$.
\end{remark}

\begin{definition}[Clique number]
We call a set of vertices $C$ a clique of $G$ if $C$ induces a complete subgraph in $G$. The \emph{clique number} $\beta$ of $G$ is the size of a largest clique of $G$.
\end{definition}

If degree of every vertex in a subgraph of $G$ is more than $d$ then it would contradict Definition \ref{def1}. Thus we have the following.

\begin{proposition}
\label{prop4}
$ \beta \leq d+1$.
\end{proposition}

\begin{remark}
\label{rem_beta}
Unlike \emph{degeneracy}, computing the \emph{clique number} of a graph is an \emph{NP-hard} problem. But, \emph{clique number} may be a better parameter for bounding running time of algorithm, if possible. To see this, consider complete bipartite graph $K_{p,q}$, where $p+q = n$. \emph{Degeneracy} of $K_{p,q}$ is $\min\{p,q\} \leq n/2$, but \emph{clique number} is only $2$.
\end{remark}

The following bound on \emph{degeneracy} is easy to show (for a proof see Appendix \ref{sec:app_prelim}).

\begin{restatable}{lemma}{LemmaDB}
\label{lemma_db}
If $G$ contains no isolated vertices, then $d +1 \leq 2 \sqrt{m}$.
\end{restatable}

\begin{corollary}
\label{cor_bb}
If $G$ contains no isolated vertices, then $\beta \leq 2 \sqrt{m}$.
\end{corollary}

If the number of cliques in a vertex clique cover of a graph $G$ is $k$, then the \emph{independence number} of $G$ is always a lower bound on $k$. Assume $G$ contains no isolated vertices. Then, any edge clique cover of $G$ is also a vertex clique cover of $G$. Thus we have the following.

\begin{proposition}
\label{prop5}
Let $G$ be a graph with no isolated vertices. If $\C= \{C_1, C_2, \ldots, C_k\}$ is an \emph{edge clique cover} of $G$, then $\alpha \leq k$.
\end{proposition}

Throughout the paper we use the following notations. Let $\langle u_1, u_2, \ldots, u_n \rangle$ be a degeneracy ordering of $V$. For each $u_i \in V$, let $N_d(u_i)$ denote the neighbours of $u_i$ in $G$ that comes later in the degeneracy ordering, i.e., $N_d(u_i) = N(u_i) \cap \{u_{i+1}, \ldots, u_n\}$. Let $N_d[u_i] = N_d(u_i) \cup \{u_i\}$. We define an ordering of edges (abbreviated \emph{DEP} for \emph{degeneracy edge permutation}) as follows.

\begin{definition}[DEP]
\label{def_DEP}
Let $E_d(x) = \{\{x,y\}| y \in N_d(x)\}$, i.e., the set of edges incident on $x$ from vertices that come after $x$ in the degeneracy ordering. Let $\pi^{E^{*}} = \langle e_1, e_2, \ldots, e_q \rangle$ be a permutation of a set of edges $E^{*}$. We say $\pi^E$ a \emph{degeneracy edge permutation} of $G$ if $\pi^E$ is a concatenation of permutations $\pi^{E_d(u_i)}$ in increasing values of $i$, i.e., $\pi^E = \langle \pi^{E_d(u_1)}, \pi^{E_d(u_2)} , \ldots, \pi^{E_d(u_{n-1})} \rangle$.
\end{definition}

\subsection{Parameterized Complexity}
\label{prelim_complexity}

We provide a brief overview of fixed-parameter tractability, kernelization, bounded search tree, and other relevant terms. For a comprehensive description, see \cite{cygan2015parameterized, downey2013fundamentals, flum2006parameterized}.

\begin{definition}[Parameterized problem]
A parameterized problem is a language $L \subseteq \Sigma^* \times \mathbb{N}$, where $\Sigma$ is a fixed and finite alphabet. For an instance $(x,p) \in \Sigma^* \times \mathbb{N}$, $p$ is called the parameter.
\end{definition}

\begin{definition}[Fixed-parameter tractability]
A parameterized problem $L \subseteq \sum^{*} \times \mathbb{N}$ is called \emph{fixed-parameter tractable} (FPT), if for any $(x,p) \in \Sigma^* \times \mathbb{N}$, we can decide whether $(x,p) \in L$ in time $f(p)|(x,p)|^{O(1)}$, where $f(p)$ is computable and is independent of $x$, but can grow arbitrarily with $p$. An algorithm is called a \emph{fixed-parameter tractable} algorithm, if it decides whether $(x,p) \in L$ in time bounded by $f(p)|(x,p)|^{O(1)}$.
\end{definition}

In a discrete optimization problem, we are interested in finding optimal value of a function $f(x)$ given that $x$ is an instance of a language $X$. Each such optimization problem corresponds to a decision problem $X_D$ that asks for existence of $x \in X$ such that $f(x) \leq p$ or $f(x) \geq p$, for a nonnegative integer $p$. A parameterized decision problem considers a decision problem $X_D$ with respect to a parameter $p^*$ such that $\mathbb{N}$ is the domain of $p^*$. If $p^*$ and $p$ are the same, i.e., the parameter is the threshold of $f(x)$ in a decision problem $X_D$ that corresponds to an optimization problem, then we call $p^*$ the \emph{natural parameterization} of $X_D$.

For convenience, the definitions of parameterized problems and fixed-parameter tractability are often described in terms of a parameter whose domain is the set of natural numbers. Without loss of generality, the definitions allow much more broader notions of parameters such as graphs, algebraic structures, vectors in $\mathbb{N}^c$ for some fixed constant $c$. An equivalent notion of \emph{FPT} is as follows (for a proof, see Appendix \ref{sec:app_prelim}).

\begin{restatable}{proposition}{PropEqvOne}
\label{prop_eqv1}
An instance $(x,p)$ of a parameterized problem is decidable in $f(p)|(x,p)|^{O(1)}$ time if and only if $(x,p)$ is decidable in $f^*(p)+|(x,p)|^{O(1)}$ time.
\end{restatable}

\begin{definition}[Equivalent instance]
Two instances $(x,p)$ and $(x^*,p^*)$ of a parameterized problem $L$ are called \emph{equivalent} if $(x,p) \in L$ if and only if $(x^*,p^*) \in L$.
\end{definition}

A natural consequence of Proposition \ref{prop_eqv1} is the notion of \emph{kernelization} also known as \emph{data reduction} or \emph{prepossessing}.

\begin{definition}[Kernelization]
\label{def_kernel}
Given an instance $(x,p)$ of a parameterized problem $L$,
\emph{kernelization} is an algorithm that in $|(x,p)|^{O(1)}$ time reduces the instance $(x,p)$  to an \emph{equivalent instance} $(x^*, p^*)$ of $L$ such that $|x^*| \leq f(p)$ and $p^* \leq g(p)$.
\end{definition}

We call the reduced instance $(x^*,p^*)$ a \emph{kernel}. If $f(p)$ is a polynomial (resp. exponential) function of $p$ then we say $L$ admits a polynomial (resp. exponential) kernel. If the reduction is to an instance of a different parameterized problem then we call the reduced instance a \emph{compression} or \emph{bi-kernel}. Kernelization often consists of a set of reduction rules such that each rule produces an \emph{equivalent instance}. We call an instance $(x,p)$ \emph{reduced} with respect to a set of rules if the rules are not applicable to $(x,p)$. 

The following equivalence is straightforward to show (see Proposition 4.8.1 of \cite{downey2013fundamentals}), which is the basis of many FPT algorithms.

\begin{proposition}
A \emph{parameterized problem} $L$ is \emph{FPT} if and only if $L$ admits a \emph{kernelization}.
\end{proposition}

Similar to polynomial time many-one reductions (also known as \emph{Karp reductions}), we have the following type of reductions for parameterized problems.

\begin{definition}[Parameterized reduction \cite{cygan2015parameterized}]
\label{def_pr}
Let $L, L^* \subseteq \Sigma^* \times \mathbb{N}$ be two parameterized problems. A parameterized reduction from $L$ to $L^*$ is an algorithm that, given an instance $(x,p)$ of $A$, returns an instance $(x^*, p^*)$ of $L^*$ such that (1) $(x,p)$ is a \emph{YES} instance of $L$ if and only if $(x^*,p^*)$ is a \emph{YES} instance of $L^*$, (2) $p^* \leq g(p)$ for some computable function $g$, and (3) runs in $f(p)|x|^{O(1)}$ time for some computable function $f$.
\end{definition}

From Definition \ref{def_pr}, if there exists a \emph{parameterized reduction} from $L$ to $L^*$ and $L^*$ has an FPT algorithm with respect to $p^*$, then $L$ also has an FPT algorithm with respect to $p$.

\textbf{Bounded search tree.} \emph{Bounded search tree} is one of the most commonly used techniques for designing FPT algorithm that often produces provably better FPT algorithms than many other techniques. The idea is to treat the sequence of decisions made by an algorithm as a tree $T$. Each node $u$ of $T$ corresponds to a choice made by the algorithm such as deciding to include an element in a partial solution. Children of a node $u$ correspond to subsequent choices that arise as ramifications of the choice made at node $u$. Each node $u$ corresponds to a branch of the tree rooted at node $u$, and the number of children of $u$ is called the \emph{branching factor} of node $u$. An algorithm systematically explores the ramifications of choices made at node $u$ by recursively visiting nodes of all branches rooted at children of $u$ until a solution is found. While visiting a child node of $u$, the algorithm makes substantial progress such as increasing the size of a partial solution. By a recurrence relation, time required in a subtree $T_u$ rooted at node $u$ can be bounded in terms of time required in the subtrees rooted at the children of $u$ and time required at node $u$ alone. The recurrence relation can be solved to provide an upper bound on the time requirement of the algorithm. A simple strategy often suffices for bounding time requirement of a search tree algorithm. Let the depth or the height of a tree $T$ be bounded by $h(p)$, and let the \emph{branching factor} of any node of $T$ be bounded by $b(p)$. Then the number of nodes in $T$ is bounded by $\{b(p)\}^{h(p)}$. If time spent at each node is bounded by a polynomial in $|(x,p)|$, then total time required by the algorithm is bounded by $\{b(p)\}^{h(p}|(x,p)|^{O(1)}$, which is an FPT running time.

\section{Enumerating Cliques of Restricted Subgraph}
\label{sec_ecrs}

In this section, we describe FPT algorithms for \emph{ECC} and \emph{ACC} (which can be extended to design algorithms for other clique cover problems such as the problems described in Section \ref{intro_problems}). Our focus is on bounded search tree algorithms for the problems, but we touch upon a number of relevant data reduction rules. The general theme of these search tree algorithms is as follows. In each node of a search tree, the corresponding algorithm branches on a set of cliques constructed from a restricted subgraph. The restriction is justified with a goal to bound the branching factors in terms of degeneracy or other parameters. A proof of \emph{NP-completeness} of \emph{ACC} is included at the end of this section.

\subsection{Edge Clique Cover}
\label{ecrs_ecc}

The exponential kernel of \cite{gramm2009data} is based on a result of \cite{gyarfas1990simple}. We use the result of \cite{gyarfas1990simple} and include an exposition of a proof of the result adapted from \cite{gramm2009data}.

\begin{definition}[Equivalent Vertices]
Let $\{x,y\}$ be an edge of $G$. $x$ and $y$ are equivalent vertices if the vertices have the same closed neighbourhood, i.e., $N[x] = N[y]$.
\end{definition}

\begin{lemma}[\cite{gyarfas1990simple}]
\label{lemma_gyarfas}
Let $G$ be a graph with $n$ vertices such that $G$ contains neither isolated vertices nor equivalent vertices. If $\C = \{C_1, C_2, \ldots, C_k\}$ is an edge clique cover of $G$, then $n < 2^k$.
\end{lemma}
\begin{proof}
For the sake of contradiction suppose $n \geq 2^k$. 
Let $\boldsymbol{B} \in {\{0,1\}}^{n \times k}$ be a matrix such that $\boldsymbol{B}[x,l] = 1$ if vertex $x$ is contained in clique $C_l$, otherwise, $\boldsymbol{B}[x, l] = 0$. The number of distinct rows in $\boldsymbol{B}$ is at most $2^k$. If it is $2^k$, then included among these rows is a row $x$ such that $\boldsymbol{B}[x,l] = 0$ for all $l \in [k]$. If the number of distinct rows is less than $2^k$, then, since $n \geq 2^k$, there exist two distinct rows $x$ and $y$ such that $\boldsymbol{B}[x,l] = \boldsymbol{B}[y,l]$ for all $l \in k$. The first condition leads to an isolated vertex, and the second condition leads to a pair of equivalent vertices. Both conditions lead to contradiction, and the claim follows.
\end{proof}

An immediate consequence of Proposition \ref{prop_dalpha} and Lemma \ref{lemma_gyarfas} is as follows.

\begin{corollary}
\label{cor_gyarfas}
Let $G$ be a graph with \emph{degeneracy} $d$ and \emph{independence number} $\alpha$ such that $G$ contains neither isolated vertices nor equivalent vertices. If $\C = \{C_1, C_2, \ldots, C_k\}$ is an edge clique cover of $G$, then $d + \alpha < 2^k$.
\end{corollary}

Let $(G,k)$ be an instance of \emph{ECC}. The following kernelization rules ensue from Lemma \ref{lemma_gyarfas}.

\begin{krule}
\label{rule_ecc_1}
If $x \in V$ is an isolated vertex, then $(G-\{x\}, k)$ and $(G,k)$ are equivalent instances. For a solution $\C$ of $(G-\{x\}, k)$, report $\C$ as the solution of $(G,k)$.
\end{krule}

\begin{krule}
\label{rule_ecc_2}
If $\{x,y\} \in E$ and $N[x] = N[y]$, then $(G-\{x\}, k)$ and $(G,k)$ are equivalent instances. For a solution $\C$ of $(G-\{x\}, k)$, report $(\C \backslash \{C_i| y \in C_i\}) \cup \{C_i \cup \{x\}| y \in C_i\}$ as the solution of $(G,k)$.
\end{krule}

Let $(G,k)$ be a reduced instance with respect to Rules \ref{rule_ecc_1} and \ref{rule_ecc_2}. We apply the following rule on $(G,k)$.

\begin{krule}
\label{rule_ecc_3}
If $n > (d+1)k$, then report $(G,k)$ is a \emph{NO} instance of \emph{ECC}.
\end{krule}

The correctness of Rule \ref{rule_ecc_3} follows from the fact that with $k$ cliques and \emph{clique number} bounded by $d+1$, we can cover at most $(d+1)k$ vertices. Since $(G,k)$ is reduced with respect to Rule \ref{rule_ecc_1}, an edge clique cover must include all the vertices of $G$.

Let $(G,k)$ be a reduced instance with respect to Rules \ref{rule_ecc_1}, \ref{rule_ecc_2}, and \ref{rule_ecc_3}. We have a \emph{kernel} with $(d+1)k$ vertices. Note that only Rules \ref{rule_ecc_1} and \ref{rule_ecc_3} are sufficient to get the \emph{kernel}. We need Rule \ref{rule_ecc_2} to show a better search tree algorithm.

The search tree algorithm of \cite{gramm2009data} (henceforth refereed as \emph{ECCG}) works by enumerating all maximal cliques that contain an uncovered edge selected at every node of a search tree. For analysis of FPT algorithm, it is sufficient to consider number of maximal cliques in a subgraph, since each such clique can be listed spending polynomial time per clique \cite{manoussakis2019new}.

\begin{lemma}[\cite{wood2011number}]
\label{lemma_mmb}
The number of maximal cliques in a graph with $n$ vertices is at most $3^{n/3}$.
\end{lemma}

Using Lemma \ref{lemma_mmb}, we can obtain a bounded search tree algorithm with depth at most $k$ and branching factor at most $3^{(d+1)k/3}$, i.e., an algorithm $A$ that takes $3^{(d+1)k^2/3}n^{O(1)}$ time. By Lemma \ref{lemma_mmb}, \emph{ECCG} takes $3^{k2^k/3}n^{O(1)}$ time. The running time of $A$ would be better (or competitive) than the running time of \emph{ECCG} if $d=O(2^k/k)$. Next, we describe a search tree algorithm with (unconditionally) better running time than \emph{ECCG}.

\begin{figure}
    \centering
    \fbox{\parbox{6.0in}{
\underline{Edge-Clique-Cover-Search$(G, k, \C)$:}

\textcolor{teal}{\\ // Abbreviated \textbf{ECCS}$(G,k,\C)$}

    \begin{enumerate}
        \item if $\C$ covers edges of $G$, then return $\C$
        \item if $k < 0$, then return $\emptyset$
        \item select the \emph{first} uncovered edge $\{x,y\}$ from the \emph{DEP} of $G$ where $y \in N_d(x)$ \textcolor{teal}{// Definition \ref{def_DEP}}
        \item for each maximal clique $Z$ containing $\{x,y\}$ in the subgraph $G[N_d[x] \cap N[y]]$ do 
    \begin{enumerate}
        \item $\Q \leftarrow \textbf{ECCS}(G, k-1, \mathcal{C} \cup \{Z\})$
        \item if $\Q \ne \emptyset$, then return $\Q$
    \end{enumerate}
    \item return $\emptyset$
    \end{enumerate} 
}}   
    \caption{A bounded search tree algorithm for \emph{ECC}, denoted \emph{ECCS}.}
    \label{fig:ecc_algo1}
\end{figure}

Figure \ref{fig:ecc_algo1} shows a bounded search tree algorithm for \emph{ECC}, henceforth refereed as \emph{ECCS}. We assume \emph{ECCS} has access to a \emph{degeneracy edge permutation} of $G$. At step 3, we select the first uncovered edge $\{x,y\}$ from the \emph{degeneracy edge permutation} such that $y \in N_d(x)$. At step 4, we restrict the search for maximal cliques in a subgraph induced by vertices of $N_d[x] \cap N[y]$ in $G$. In each iteration of step 4, we select a maximal clique from the restricted search space, and cover the edge $\{x,y\}$ with the clique and recurse. This process is repeated until a solution is found in a branch (step 4b), or all branches are exhausted without any solution (step 5).

We point out the crucial differences between \emph{ECCG} and \emph{ECCS}. At step 3, \emph{ECCG} uses a heuristic to select an edge, for which \emph{ECCG} has to enumerate maximal cliques of the subgraph $G[N[x] \cap N[y]]$, in contrast to enumerating maximal cliques of a restricted subgraph such as $G[N_d[x] \cap N[y]]$. Consequently, \emph{ECCG} can only provide an upper bound of $3^{2^k/3}$ on the number of maximal cliques enumerated at step 4: \emph{kernelization}, based on Lemma \ref{lemma_gyarfas}, can only ensure that the number of vertices are bounded by $2^k$.

Next, we describe correctness and running time of \emph{ECCS}. From the choices made at step 3, we observe following invariant maintained by \emph{ECCS}.

\begin{proposition}
\label{inv_eccs}
Let $\langle u_1, u_2, \ldots, u_n \rangle$ be a degeneracy ordering of $V$. At any node of a search tree of \emph{ECCS}, if $x=u_i$, then all edges incident on vertices of $\{u_1, u_2, \ldots, u_{i-1}\}$ are covered.
\end{proposition}

We also note a property that holds for induced subgraphs, but does not necessarily hold for arbitrary subgraphs.

\begin{proposition}
\label{inv_ecc}
For a graph $G=(V,E)$, let $X \subseteq V$ be a set of vertices.
If $(G,k)$ is a \emph{YES} instance of \emph{ECC}, then $(G[X],k)$ is a \emph{YES} instance of \emph{ECC}. 
\end{proposition}

\begin{lemma}
\label{lemma_ecc_proof}
Algorithm \emph{ECCS} correctly solves the parameterized problem \emph{ECC}.
\end{lemma}
\begin{proof}
Let $\A$ be a family of algorithms that consider uncovered edges $\{x,y\}$ in arbitrary order in step 3 of \emph{ECCS}, and in step 4 branch into each of the choices of maximal cliques containing \{x,y\} in the subgraph induced by $N[x] \cap N[y]$. 

Consider any algorithm $A_i \in A$. For a search tree $T$ of $A_i$, let us define a node $u^T$ to be a \emph{YES} node if $A_i$ returns at step 1 of $u^T$. We claim that $(G,k)$ is a \emph{YES} instance of \emph{ECC} if and only if $A_i$ returns from a \emph{YES} node. If $A_i$ returns from a \emph{YES} node, then clearly $(G,k)$ is a \emph{YES} instance of \emph{ECC}. It remains to show that if $(G,k)$ is a \emph{YES} instance of \emph{ECC}, then $A_i$ would return from a \emph{YES} node. This can be seen as follows. $A_i$ considers all possible ways to cover every uncovered edge: $\{x,y\}$ can be covered with one of the maximal cliques containing $\{x,y\}$ in the subgraph induced by $N[x] \cap N[y]$. $A_i$ considers all such possible maximal cliques, and allows a search tree containing at most $k$ cliques in any leaves of the search tree.

\emph{ECCS} deviates from the family of algorithms $\A$. Next, we show that the deviations are safe. Clearly, if \emph{ECCS} returns from a node at step 1, then $(G,k)$ is a \emph{YES} instance of \emph{ECC}. Therefore, we only need to show that if any algorithm $A_i \in \A$ returns from a \emph{YES} node, then \emph{ECCS} would also return from a node at step 1.

At each node of the search tree of \emph{ECCS}, $y\in N_d[x]$, and  $N_d[x] \subseteq \{u_i, u_{i+1}, \ldots, u_n\}$. Therefore, $N_d[x] \cap N[y] \subseteq \{u_i, u_{i+1}, \ldots, u_n\}$. In the enumeration of maximal cliques, vertices that are not considered by \emph{ECCS} are $(N[x] \cap N[y]) \backslash (N_d[x] \cap N[y]) = (N(x) \backslash N_d(x)) \cap N(y) \subseteq N(x) \backslash N_d(x) \subseteq \{u_1, u_2, \ldots, u_{i-1}\}$.

By Proposition \ref{inv_eccs}, all edges incident on vertices of $\{u_1, u_2, \ldots, u_{i-1}\}$ are covered when we are considering maximal cliques for an edge $\{x,y\}$. By Proposition \ref{inv_ecc}, if $(G[N[x] \cap N[y]],k^{*})$ is a \emph{YES} instance of \emph{ECC}, then $(G[N_d[x] \cap N[y]],k^{*})$ is a \emph{YES} instance of \emph{ECC}. Therefore, it is safe to consider maximal cliques for an edge $\{x,y\}$ only in the subgraph $G[N_d[x] \cap N[y]]$. It follows that if any algorithm $A_i \in \A$ returns from a \emph{YES} node, then \emph{ECCS} would also return from a node at step 1.
\end{proof}

\begin{lemma}
\label{lemma_eccnn}
The number of nodes in a search tree of \emph{ECCS} is at most $1.4423^{dk}$.
\end{lemma}
\begin{proof}
The number of vertices in a subgraph induced by vertices of $N_d[x] \cap N[y]$ is at most $d+1$ i.e., $|N_d[x] \cap N[y]| \leq d+1$. Since the vertices $x$ and $y$ are fixed, we only need to enumerate maximal cliques of the subgraph $G[N_d(x) \cap N(y)]$, containing at most $d-1$ vertices.

By Lemma \ref{lemma_mmb}, at each node of the search tree, the number of maximal cliques containing $x$ and $y$ in the subgraph induced by vertices of $N_d[x] \cap N[y]$ is at most $3^{(d-1)/3} < 3^{d/3}$. Therefore, the total number of nodes in the search tree is at most $3^{dk/3} \leq 1.4423^{dk}$
\end{proof}

From Figure \ref{fig:ecc_algo1}, it is evident that the time spent at every node of a search tree of \emph{ECCS} is bounded by a polynomial of input size. Considering the time needed for data reduction, we obtain Theorem \ref{thm_ecc} from Lemma \ref{lemma_eccnn} and Proposition \ref{prop_eqv1}.

Next, we discuss a number of consequences of \emph{ECCS}.

Using Lemma \ref{lemma_mmb}, the number of nodes is a search tree of \emph{ECCG} is bounded by $3^{k2^k/3}$. By Corollary \ref{cor_gyarfas}, $2^k > d+\alpha$. Therefore, we have a factor of $\frac{3^{k2^k/3}}{3^{dk/3}} > \frac{3^{k(d+\alpha)/3}}{3^{dk/3}} = 3^{k\alpha/3} = 2^{\Omega(k\alpha)} = 2^{\Omega(\alpha^2)}$ improvement over the bound on search tree size. The last equality follows from Lemma \ref{prop5}. The improvement becomes significant for sparse graphs, such as when $d = o(2^k)$ or $\alpha = \Omega(n)$, and we may obtain a factor of $2^{2^{O(k)}}$ reduction on the number of nodes in the search tree of \emph{ECCS} compared to \emph{ECCG}. Note that with kernelization we have $n < 2^k$. Therefore, preceding improvement carry over to the overall running time.

We point out that single exponential dependence on $k$, with low degree polynomial of $k$ in the exponent, suffices for sparse graphs.

\begin{corollary}
For $d = O(k)$, \emph{ECC} has an FPT algorithm running in $2^{O(k^2)}n^{O(1)}$ time.
\end{corollary}

\begin{corollary}
\label{cor_eccs}
For $k = \Omega(\sqrt{m})$, \emph{ECC} has an FPT algorithm running in $2^{O(k^2)}n^{O(1)}$ time.
\end{corollary}

Corollary \ref{cor_eccs} follows from Lemma \ref{lemma_db}.

From the preceding two observations, it is obvious that the inherent difficulty of solving \emph{ECC} on sparse graphs is poorly captured by parameter $k$ alone. Next, we demonstrate that this difficulty is indeed better captured by degeneracy.

\begin{corollary}
\label{cor_eccs2}
If $m \geq \lfloor \frac{n^2}{4} \rfloor$, then running time of \emph{ECCS} is $2^{O(d^3)}n^{O(1)}$.    
\end{corollary}
\begin{proof}
For \emph{ECC}, we can assume $k \leq \lfloor \frac{n^2}{4} \rfloor$ \cite{erdos1966representation}. Also, by Proposition \ref{prop2}, we can assume $k \leq m < nd$. Since $\lfloor \frac{n^2}{4} \rfloor \leq m$, combining preceding inequalities, we have $k < 4 d^2$. Therefore, \emph{ECCS} takes $2^{O(d^3)}n^{O(1)}$ time.
\end{proof}

By algebraic manipulation similar to the preceding proof, we can show the following.

\begin{corollary}
\label{cor_eccs3}
For $c >0, \epsilon > 0$, if $m \geq cn ^{1+\epsilon}$, then running time of \emph{ECCS} is $2^{O(d^2(d/c)^{1/\epsilon})}n^{O(1)}$.
\end{corollary}

\begin{corollary}
\label{cor_eccs4}
For $c >0, \epsilon > 0$, if $m \geq cn(\log n)^{\epsilon}$, then running time of \emph{ECCS} is $2^{O( d^22^{(d/c)^{1 / \epsilon}})}n^{O(1)}$.
\end{corollary}

The preceding observations show that the dependence of the running time of \emph{ECCS} on \emph{degeneracy} shifts from single exponential to double exponential as graph becomes sparser. Unlike double exponential dependence on $k$, this shift captures what we expect for sparse graphs.

\emph{ECC} is \emph{NP-complete} on planar graphs \cite{chang2001tree}. Degeneracy of planar graph is at most 5. Thus we have the following.

\begin{corollary}
\label{cor_eccs5}
For planar graphs, \emph{ECC} has an FPT algorithm running in $2^{O(k)}n^{O(1)}$ time.
\end{corollary}

Note that for planar graphs we can easily obtain a \emph{kernel} with $4k$ vertices, since \emph{clique number} of planar graph is at most $4$. Then, instead of running time with exponent linear in $k$ (as in Corollary \ref{cor_eccs2}), we would get running time with exponent quadratic in $k$.

\subsection{Assignment Clique Cover}
\label{ecrs_acc}

We describe an FPT algorithm for \emph{ACC}. First, we describe a set of simple data reduction rules that reduce a given instance $(G,t)$ of \emph{ACC}. We note that many data reduction rules for \emph{ECC} may not be directly applicable for \emph{ACC}, for example Rule \ref{rule_ecc_2} (we expand on this later).

\begin{krule}
\label{rule_acc_1}
If $x \in V$ is an isolated vertex or a vertex with no uncovered incident edges, then $(G-\{x\}, t)$ and $(G,t)$ are equivalent instances. For a solution $\C$ of $(G-\{x\}, t)$, report $\C$ as the solution of $(G,t)$.
\end{krule}

\begin{krule}
\label{rule_acc_2}
Let $x$ be vertex with all its incident edges uncovered, and $N[x]$ induces a clique $C$ in $G$. Mark all edges of $C$ as covered. Then $(G-\{x\}, t- |N[x]|)$ and $(G,t)$ are equivalent instances. For a solution $\C^{*}$ of $(G-\{x\}, t- |N[x]|)$, report $\C^{*} \cup \{C\}$ as the solution of $(G,t)$.
\end{krule}

Rule \ref{rule_acc_2} is correct. Every edge $\{x,y\} \in E$ has to be covered by a clique. Therefore, covering the edge $\{x,y\}$ with a clique other than the clique induced by $N[x]$ would increase number of assignments of $x$.

\begin{krule}
\label{rule_acc_3}
Let $e=\{x,y\}$ be an uncovered edge, and $N(x) \cap N(y) = \emptyset$. Mark the edge $e$ as covered. Then $(G-\{e\}, t-2)$ and $(G,t)$ are equivalent instances. For a solution $\C^{*}$ of $(G-\{e\}, t-2)$, report $\C^{*} \cup \{\{x,y\}\}$ as the solution of $(G,t)$.
\end{krule}

Rule \ref{rule_acc_3} is correct: edge $\{x,y\}$ can only be covered by a clique $C=\{x,y\}$. Let $(G,t)$ be a reduced instance with respect to Rules \ref{rule_acc_1}, \ref{rule_acc_2}, and \ref{rule_acc_3}. We apply following rule on $(G,t)$.

\begin{krule}
\label{rule_acc_4}
If $n > t$, then report $(G,t)$ is a \emph{NO} instance of \emph{ACC}.
\end{krule}

Since $(G,t)$ is reduced with respect to Rule \ref{rule_acc_1}, an \emph{assignment clique cover} must include all the vertices of $G$. Therefore, Rule \ref{rule_acc_3} is correct. Next, we describe a bounded search tree algorithm for \emph{ACC}.

\begin{figure}
    \centering
    \fbox{\parbox{6.0in}{
\underline{Assignment-Clique-Cover-Search$(G, t, \C)$:}

\textcolor{teal}{\\ // Abbreviated \textbf{ACCS}$(G,t,\C)$}

    \begin{enumerate}
        \item if $\C$ covers edges of $G$, then return $\C$
        \item if $t < 2$, then return $\emptyset$
        \item select the \emph{first} uncovered edge $\{x,y\}$ from the \emph{DEP} of $G$ where $y \in N_d(x)$ \textcolor{teal}{// Definition \ref{def_DEP}}
        \item for each \emph{clique} $Z$ containing $\{x,y\}$ in the subgraph $G[N_d[x] \cap N[y]]$ do
    \begin{enumerate}
        \item $\Q \leftarrow \emptyset$
        \item if $t - |Z| \geq 0$, then $\Q \leftarrow \textbf{ACCS}(G, t-|Z|, \mathcal{C} \cup \{Z\})$
        \item if $\Q \ne \emptyset$, then return $\Q$
    \end{enumerate}
    \item return $\emptyset$
    \end{enumerate} 
}}   
    \caption{A bounded search tree algorithm for \emph{ACC}, denoted \emph{ACCS}.}
    \label{fig:acc_algo1}
\end{figure}

Figure \ref{fig:acc_algo1} shows a bounded search tree algorithm for \emph{ACC}, henceforth referred as \emph{ACCS}. \emph{ACCS} uses the same strategy as \emph{ECCS} to select uncovered edge at step 3. Unlike \emph{ECCS}, \emph{ACCS} searches for \emph{cliques} rather than \emph{maximal cliques} at step 4. But, the restriction of search space for cliques is same as \emph{ECCS}: subgraph induced by vertices of $N_d[x] \cap N[y]$ in $G$.

The need for enumerating all cliques in step 4, instead of only maximal cliques, follows from the fact that a maximal clique may include additional vertices not required for an \emph{edge clique cover}, and the corresponding branch of a search tree may miss a solution due to excess decrement of the value of parameter $t$. The correctness of subgraph restriction is similar to what has been described in the proof of Lemma \ref{lemma_ecc_proof}. Therefore, we conclude the following.

\begin{lemma}
Algorithm \emph{ACCS} correctly solves the parameterized problem \emph{ACC}.
\end{lemma}

\begin{theorem}
\label{thm_acc3}
\emph{ACC} has an FPT algorithm running in $1.4143^{t^2}n^{O(1)}$ time.
\end{theorem}
\begin{proof}
As we have shown in the proof of Lemma \ref{lemma_eccnn}, at each node of a search tree of \emph{ACCS}, we only need to enumerate cliques of the subgraph $G[N_d(x) \cap N(y)]$ with at most $d-1$ vertices. At each node of a search tree, the number of cliques containing the edge $\{x,y\}$ in the subgraph $G[N_d[x] \cap N[y]]$ is at most $\sum_{s=1}^{d-1} \binom{d-1}{s} = 2^{d-1}$. The depth of a search tree is bounded by $t/2$. Therefore, the total number of nodes in a search tree is bounded by $2^{dt/2} \leq 1.4143^{t^2}$, since by Rule \ref{rule_acc_4}, $d < n \leq t$. The claim follows from Proposition \ref{prop_eqv1}, considering time needed for data reduction.
\end{proof}

\begin{remark}
For a reduced instance of \emph{ACC} with respect to Rule \ref{rule_acc_1}, $n \leq t$. Therefore, to achieve a running time with respect to parameter $t$, it would suffice to select edges arbitrarily at step 3 of \emph{ACCS}, and enumerate cliques of subgraph $G[N[x] \cap N[y]]$ at step 4. The restriction of search for cliques in a subgraph $G[N_d[x]\cap N[y]]$ is what makes the algorithm attuned to sparse graphs.
\end{remark}

\begin{corollary}
\label{cor_acc2}
An \emph{assignment-minimum clique cover} of a graph $G$ can be found in $2^{O(mn)}$ time using $O(m^{3/2})$ space.
\end{corollary}
\begin{proof}
We can invoke the FPT algorithm for \emph{ACC} that we just have described with values of $t$ in the range $[1, 2m]$ (by linearly increasing values of $t$ or doing a binary search in the range) until a solution is found. Total running time would be $\sum_{t \in [2m]} 2^{O(dt)}n^{O(1)} =2^{O(dm)} = 2^{O(mn)}$.

The depth a search tree of $\emph{ACCS}$ is at most $m$, and each such clique would take $O(d)$ space. Also, enumeration of cliques in step 4 can be done using a binary vector of size $O(d)$. Therefore, the total space usage of \emph{ACCS} is $O(dm) = O(m^{3/2})$. The last equality follows from Lemma \ref{lemma_db}.
\end{proof}

We have noted that Rule \ref{rule_ecc_2} is not directly applicable for \emph{ACC}, but we can use rules like that in a broader context. The following two data reduction rules are applicable for computing an \emph{assignment-minimum clique cover} of a graph $G$.

\begin{krule}
\label{rule_amcc1}
Let $\{x,y\}$ be an edge of $G$ and $N[x]=N[y]$. If $\C$ is an \emph{assignment-minimum clique cover} of $G \backslash \{x\}$, then $(\C \backslash \{C_l | y \in C_l\}) \cup \{C_l \cup \{x\} | y \in C_l\}$ is an \emph{assignment-minimum clique cover} of $G$.
\end{krule}

\begin{krule}
\label{rule_amcc2}
Let $x$ and $y$ be distinct vertices of $G$ such that $\{x,y\} \not\in E$ and $N(x)=N(y)$. If $\C$ is an \emph{assignment-minimum clique cover} of $G \backslash \{x\}$, then $\C \cup \{C_l\backslash\{y\} \cup \{x\} | y \in C_l, C_l \in \C\}$ is an \emph{assignment-minimum clique cover} of $G$.
\end{krule}

Since a corresponding correct parameter $t$ for $G\backslash \{x\}$ cannot be known beforehand, there is no equivalent of Rule \ref{rule_amcc1} or Rule \ref{rule_amcc2} for \emph{ACC}. It is straightforward to see that the preceding two rules are correct.

\subsection{Other Clique Covers}

What we have shown in the preceding discussion for \emph{ECC} and \emph{ACC} can be extended to design algorithms for other clique cover problems, namely \emph{WECP}, \emph{EWCD}, and \emph{LRVCC}. We refrain from going into the details of these extensions. Instead, we describe algorithms for these problems using a different framework in Section \ref{sec_algos}: one may utilize the descriptions in Section \ref{sec_algos} to design full-fledged algorithms based on enumeration of cliques of restricted subgraph.

\subsection{\emph{NP-completeness} of \emph{ACC}}

We conclude this section with a proof of Theorem \ref{thm_acc}. We use a construction used by \cite{kou1978covering} (Proposition 2).

\begin{proof}[Proof of Theorem \ref{thm_acc}]
\emph{ACC} is in \emph{NP}: given an instance $(G,t)$ and a certificate $\C$, in polynomial time we can verify that $\C$ covers edges of $G=(V,E)$, and $\sum_{C_l \in \C} |C_l| \leq t$.

Let $(G^{*},k)$ be an instance of \emph{VCC}. We construct an instance $(G,t)$ of $ACC$ from $(G^{*},k)$, and show that $(G^{*},k)$ is a \emph{YES} instance of \emph{VCC} if and only if $(G, t)$ is a \emph{YES} instance of \emph{ACC}.

Let $G^{*} = (V^{*},E^{*})$ has $n$ vertices and $m$ edges. To construct $G$, we start with an empty graph and include all the vertices and edges of $G^{*}$. Then, we include a set of additional vertices $X= \{x_1, x_2, \ldots, x_q\}$ with $q \geq 2m+1$. Then, we connect each pair of vertices $x_i \in X$ and $v \in V^{*}$ with an edge. Resulting graph $G=(V,E)$ has vertices $V= V^{*} \cup X$ and edges $E= E^{*} \cup \{\{x_i,v\} | x_i \in X, v \in V^{*}\}$. We complete the construction by setting $t = (n+k)q+2m$.

\textbf{Only if.} Let $(G^{*},k)$ be a \emph{YES} instance of \emph{VCC}, and $\C^{*} = \{C_1^{*}, C_2^{*}, \ldots, C_k^{*}\}$ be a corresponding \emph{vertex clique cover} of $G^{*}$. Since each vertex $x_i$ is connected to all the vertices of $V^{*}$ and no pair of vertices in $X$ are connected by edge, we can cover all the edges incident on $x_i$ with a set of cliques $\C({x_i}) = \{C_l^{*} \cup \{x_i\} | l \in [k]\}$. Since the number of individual assignments of vertices to cliques in $\C^{*}$ is $n$, the number of individual assignments of vertices to cliques in $\C(x_i)$ is $n+k$. Therefore, to cover the edges incident on a vertex $x_i$, we would need at most $\sum_{C_l \in \C({x_i})} |C_l| \leq n+k$ individual assignments of vertices to cliques. To cover all the edges incident on the vertices of $X$, we would need at most $\sum_{x_i \in X} \sum_{C_l \in \C({x_i})} |C_l| \leq (n+k)q$ individual assignments of vertices to cliques. Finally, to cover $m$ edges of $G^{*}$, we would need at most $2m$ individual assignments of vertices to cliques. Therefore, in polynomial time, we can construct an \emph{edge clique cover} $\C$ of $G$ such that $\sum_{C_l \in \C} |C_l| \leq  (n+k)q + 2m = t$. It follows that $(G,t)$ is a \emph{YES} instance of \emph{ACC}.

\textbf{If.} Let $(G,t)$ be a \emph{YES} instance of \emph{ACC}, and $\C$ be a corresponding  \emph{assignment clique cover} of $G$. Let $\C(x_i) = \{C_l | x_i \in C_l, C_l \in \C\}$. For any $i \in [q]$, the set of cliques $\C(x_i)$ with $x_i$ removed is a \emph{vertex clique cover} of $G^{*}$. We need to choose $i \in [q]$ such that $|\C(x_i)| \leq k$. Note that for any pair of distinct vertices $x_i$ and $x_j$ of $X$, $\C(x_i)$ and $\C(x_j)$ are disjoint, since $\{x_i$, $x_j\} \not\in E$. To show $(G^{*},|\C(x_i)|)$ is a \emph{YES} instance of \emph{VCC}, choosing $\C(x_i)$ such that $\sum_{C_l \in \C(x_i)} |C_l|$ is minimum suffices. We have

$$ \sum_{C_l \in \C(x_i)} |C_l| = \min_{j \in [q]} \sum_{C_l \in \C(x_j)} |C_l| \leq \frac{\sum_{j \in [q]}\sum_{C_l \in \C(x_j)} |C_l|}{q} \leq \frac{\sum_{C_l \in \C}|C_l|}{q} \leq \frac{t}{q}$$ $$  = \frac{(n+k)q+2m}{q} =n+k +\frac{2m}{q} \leq n+k .$$
 
The last inequality follows from the choice of $q \geq 2m+1$. For the sake of contradiction, assume $|\C(x_i)| = k^{*} > k$. Then, by the definition of $\C(x_i)$, $x_i$ appears in $k^{*}$ times in $\C(x_i)$. Since $x_i$ is connected to all $n$ vertices of $V^{*}$, $\sum_{C_l \in \C(x_i)} |C_l| \geq n+k^{*} > n+k$, which contradicts the fact that $\sum_{C_l \in \C(x_i)} |C_l| \leq n+k$. Therefore, $|\C(x_i)| \leq k$, i.e., $(G^{*}, |\C(x_i)|)$ is a \emph{YES} instance of \emph{VCC}.
\end{proof}

\section{New Framework: Building Blocks}
\label{sec_bblock}

We describe underlying concepts of a new framework that we use in Section \ref{sec_algos} for algorithm design. Our framework is based on a relaxation of a global minimality of clique cover and a set representation of graph. We include a number of characterizations which may be of independent interest (could be used in other contexts such as designing polynomial time algorithm for specific class of graphs). We start with a discussion of a number of pitfalls that we want to deal with the new framework.

\subsection{Pitfalls Addressed}

In terms of exact algorithm design, we have found two major ways that have been adopted to systematically explore search space of clique covers of a graph. One way is to explore matrices $\mathbf{B} \in \{0,1\}^{n \times m}$ such that $\mathbf{B}\mathbf{B^T} = \mathbf{A}$ where $\mathbf{A}$ is the adjacency matrix of graph. The matrix multiplication is defined with addition rule $1+1=1$, and diagonal entries of $\mathbf{A}$ are all assumed to be $1$. Several FPT algorithms \cite{feldmann2020fixed, cooley2021parameterized} are based on exploration of this type of search space, where binary matrices of dimension $O(k) \times O(k)$ are enumerated from a search space of size $2^{O(k^2 \log k)}$. Consequently, this type of algorithms does not have any regard for sparsity of graph, making them unsuitable even for moderately small values of parameters.

The other major way is to explore search space consisting of maximal cliques of (sub)graph. Depending on objective of clique cover, exploring search space of maximal cliques could become infeasible, even for moderately small input size. An example is the search tree algorithm of \cite{ennis2012assignment} for computing \emph{assignment-minimum clique cover} of a graph. Search tree size of \cite{ennis2012assignment} is double exponential in input size in the worst case, and the algorithm also requires space exponential in input size in the worst case.

Algorithms described in Section \ref{sec_ecrs} explore search spaces that are structurally dependent on the \emph{degeneracy} of graph. We ask whether a more robust parameter than \emph{degeneracy} can be used to restrict search space further. More specifically, we ask whether it is possible to restrict search space structurally with respect to the \emph{clique number} of graph. Recall \emph{clique number} is always bounded by \emph{degeneracy} (Proposition \ref{prop4}). Moreover, \emph{degeneracy} can grow linearly with the input size, while \emph{clique number} remains constant (Remark \ref{rem_beta}).

In this section, our discussion would be focused on \emph{edge clique cover} and number of cliques in the cover. We use following definitions in the subsequent discussion.

\begin{definition}[Minimum Clique Cover]
A \emph{minimum clique cover} of a graph $G$ is an \emph{edge clique cover} of $G$ with smallest number of cliques.
\end{definition}

\begin{definition}[Minimal Clique Cover]
A \emph{minimal clique cover} $\C$ of a graph $G$ is an \emph{edge clique cover} of $G$ such that no other \emph{edge clique cover} of $G$ is contained in $\C$, i.e., $\C$ an inclusion-wise minimal set of cliques that covers all the edges of $G$.
\end{definition}

It is evident from the preceding two definitions that a \emph{minimum clique cover} is also a \emph{minimal clique cover}. Therefore, it is natural to ask whether we can find a \emph{minimum clique cover} by systematically exploring a search space of \emph{minimal clique covers}. We face with a number of difficulties.

First, note that the complete graph $K_n$ has a \emph{minimal clique cover} of size $\binom{n}{2}$, clearly indicating a difficult search space to explore. Second, it is not clear how to systematically explore the search space of \emph{minimal clique covers}. Third, any conceivable way to systematically explore the set of \emph{minimal clique covers} of a graph would require a massive amount of tests for minimality and provision for removing and restoring cliques, i.e., enumerating \emph{minimal clique covers} of a graph would be prohibitive in terms of time and space. Next, we introduce a relaxation on global minimality to avoid these pitfalls.

\subsection{Locally Minimal Clique Cover}

\begin{definition}[Locally Minimal Clique Cover]
\label{def_lmcc}
Let $H=(V_H, E_H)$ be a proper subgraph of a graph $G=(V,E)$, and $\C= \{C_1, C_2, \ldots, C_k\}$ be an \emph{edge clique cover} of $H$. Let $\{x,y\}$ be an uncovered edge of $G$, and $H^{*}$ be the subgraph induced by vertices of $V_H \cup \{x,y\}$ in $G$. Let $\C^{*}$ be an \emph{edge clique cover} of a subgraph of $H^{*}$ constructed as follows.

(i) If the edge $\{x,y\}$ can be covered in a clique $C_l \in \C$, then we cover $\{x,y\}$ with exactly one such clique of $\C$ by letting $C_l^{*} = C_l \cup \{x,y\}$, and we set $\C^{*}$ to be $(\C \backslash \{C_l\}) \cup \{C_l^{*}\}$.

(ii) Otherwise, we create a new clique $C_{k+1} = \{x,y\}$, and set $\C^{*}$ to be $\C \cup \{C_{k+1}\}$.

We call an \emph{edge clique cover} $\C$ \emph{locally minimal} if $\C$ is obtained from an empty clique cover using aforesaid construction; i.e., for every expansion of the cliques contained in $\C$, we either have used (i) whenever applicable, or (ii) otherwise.
\end{definition}

The bottom-up constructive clique cover in Definition \ref{def_lmcc} is simple but has far-reaching consequences. First, search space of \emph{locally minimal clique cover} is easy to explore systematically. Second, the search space is much more compact as it does not try to construct maximal cliques, let alone all maximal cliques of graph. Third, we will be able to make the search space exploration efficient, with the help of a set representation introduced later. Forth, it admits a number of desirable characterizations. We elaborate on these in subsequent discussion.

\begin{definition}[$E_{\pi(i)}$]
\label{def_epii}
Let $\C=\{C_1, C_2, \ldots, C_k\}$ be an \emph{edge clique cover} of a graph $G=(V,E)$. Let $\pi$ be a permutation of $[k]$. $E_{\pi(i)} = \{\{x,y\}\in C_{\pi(i)} | \{x,y\} \not\in C_{\pi(j)}, j <i\}$, i.e., the set of edges exclusive to $C_{\pi(i)}$ with respect to the cliques $\{C_{\pi(1)}, C_{\pi(2)},\ldots,C_{\pi(i-1)}\}$.
\end{definition}

An immediate observation for \emph{locally minimal clique cover} is as follows.

\begin{proposition}
\label{prop_lmcc}
If $\C=\{C_1, C_2, \ldots, C_k\}$ is a \emph{locally minimal clique cover} of $G$, then there exists a permutation $\pi$ such that $E_{\pi(i)} \ne \emptyset$, for all $i \in [k]$.
\end{proposition}

The following is a characterization that \emph{minimal clique cover} does not admit, but \emph{locally minimal clique cover} does. Recall the \emph{minimal clique cover} of $K_n$ with $\binom{n}{2}$ cliques, which arises precisely due to the fact that \emph{minimal clique cover} lacks the following characterization.

\begin{proposition}
\label{prop_lmcc2}
Let $\C=\{C_1, C_2, \ldots, C_k\}$ be an \emph{edge clique cover} of $G=(V,E)$. If $\C$ is \emph{locally minimal}, then for each pair of distinct cliques $\{C_a, C_b\}$ in $\C$, there exist $x \in C_a$, $y \in C_b$ such that $x \ne y$ and $\{x,y\} \not\in E$.
\end{proposition}
\begin{proof}
We prove the contrapositive, i.e., if for some pairs of distinct cliques $\{C_a, C_b\}$ in $\C$ it holds that for all $x \in C_a$ and for all $y \in C_b$, $x\ne y$, $\{x,y\} \in E$, then $\C$ is not a \emph{locally minimal clique cover} of $G$.

Let $\pi$ be a permutation of $[k]$ and let $\pi(i) = a$ and $\pi(j) = b$. WLOG assume $i < j$, and consider an edge $\{y,z\} \in E_{\pi(j)}$. If no such edge $\{y, z\}$ exists for all $\pi$, then, by Proposition \ref{prop_lmcc}, $\C$ is not \emph{locally minimal}. By our assumption, for all $x\in C_a\backslash \{y,z\}$, $\{x,z\} \in E$ and $\{x,y\} \in E$. Clearly, rule (i) of the construction of \emph{locally minimal clique cover} is applicable to the edge $\{y,z\}$, but rule (ii) of the construction is being applied. Since preceding argument holds every permutation $\pi$, it follows that $\C$ is not \emph{locally minimal}.
\end{proof}

The converse of Proposition \ref{prop_lmcc2} does not hold. But, we can establish a close equivalent of the converse. This will be crucial for showing sufficiency of search for different clique covers in a search space of \emph{locally minimal clique covers}.

\begin{proposition}
\label{prop_lmcc3}
Let $\C=\{C_1, C_2, \ldots, C_k\}$ be an \emph{edge clique cover} of $G=(V,E)$. If for each pair of distinct cliques $\{C_a, C_b\}$ in $\C$, there exist $x \in C_a$, $y \in C_b$ such that $x \ne y$ and $\{x,y\} \not\in E$, then there exists a permutation $\pi$ of $\{1,\ldots,k\}$ such that $E_{\pi(i)} \ne \emptyset$, for all $i \in [k]$.
\end{proposition}
\begin{proof}
Let $P_1$ denote the property that for each pair of distinct cliques $\{C_a , C_b\}$ in $\C$, there exist $x\in C_a$, $y\in C_b$ such that $x\ne y$ and $\{x,y\} \not\in E$. Let $P_2$ denote the property that for all permutation $\pi$ of $\{1, 2, \ldots, k\}$ there exists $i \in [k]$ such that $E_{\pi(i)} = \emptyset$. We want to show $P_1 \implies \neg P_2$, equivalently $P_2 \implies \neg P_1$.

Consider any $C_j$ such that $\pi(k) = j$ and $E_{\pi(k)} \ne \emptyset$. We remove $C_j$ from further consideration by considering a permutation $\pi$ such that $\pi(k) = j$. We restrict our permutation space to $\{1,\ldots, k-1\}$ from $\{1,\ldots, k\}$ as depicted below.

\begin{table}[!hbtp]
    \centering
    \begin{tabular}{|c|c|c|c||c|}
    \hline
         1 & 2 & \ldots & $ k-1$ & $k$ \\
         \hline
         $C_{\pi(1)}$ & $C_{\pi(2)}$ &\ldots  & $C_{\pi(k-1)}$ & $C_j$\\
         \hline
    \end{tabular}
\end{table}

Repeatedly removing all such $C_j$ depicted above from further consideration, we would be able to restrict our permutation space to the permutations of $\{1,\ldots,l\}$ such that $E_{\pi(l)} = \emptyset$, for any permutation $\pi$ of $\{1,\ldots,l\}$. Note that $l > 1$ since $E_{\pi(1)} \ne \emptyset$ for all $\pi$. Following two cases completely characterize the restricted permutation space.

\begin{enumerate}
    \item For $i , j \in [l]$ there exists a pair of distinct cliques $\{C_i, C_j\}$ such that $C_i = C_j$.
    \item A least one non-trivial covering of some clique $C_g$ of $G$ exists in $\{C_1,\ldots, C_l\}$ (a covering of a clique is trivial if it consists of a single clique).
\end{enumerate}

It is straightforward to see that if (1) holds, then $P_1$ is false. If (2) holds, then any pairs of cliques in the non-trivial covering makes $P_1$ false.
\end{proof}

Next characterization shows a close equivalent of the converse of Proposition \ref{prop_lmcc2}.

\begin{proposition}
\label{prop_lmcc4}
Let $\C=\{C_1, C_2, \ldots, C_k\}$ be an \emph{edge clique cover} of $G=(V,E)$. If for each pair of distinct cliques $\{C_a, C_b\}$ in $\C$, there exist $x \in C_a$, $y \in C_b$ such that $x \ne y$ and $\{x,y\} \not\in E$, then there exists a \emph{locally minimal clique cover} $\C^{*}$ of $G$ such that $|\C|= |\C^{*}|$ and $C^{*}_l \subseteq C_l$, for all $l \in [k]$.
\end{proposition}
\begin{proof}
From Proposition \ref{prop_lmcc3}, we have a permutation $\pi$ of $\{1,\ldots,k\}$ such that $E_{\pi(i)} \ne \emptyset$ for all $i \in [k]$. Consider any such $\pi$, and let $C^{*}_{\pi(i)} = \{x | \{x,y\} \in E_{\pi(i)}\}$. Clearly, $C^{*}_{\pi(i)} \subseteq C_{\pi(i)}$ for all $i \in [k]$, and by definition of $E_{\pi(i)}$, $\C^{*}$ is an \emph{edge clique cover} of $G$ with size $|\C|$.
\end{proof}

Following characterization shows that it is sufficient to search for \emph{minimum clique cover} in a search space of \emph{locally minimal clique cover}.

\begin{proposition}
\label{prop_lmcc5}
Let $\C= \{C_1, C_2, \ldots, C_k\}$ be an \emph{edge clique cover} of a graph $G=(V,E)$. If $\C$ is a \emph{minimum clique cover} of $G$, then there exists a \emph{locally minimal clique cover} $\C^{*}$ of $G$ such that $|\C^{*}| = |\C|$ and $C^{*}_l \subseteq C_l$ for all $l \in [k]$.
\end{proposition}
\begin{proof}
Any \emph{minimum clique cover} satisfies the property that for each pair of distinct cliques $\{C_i, C_j\} \in \C$, there exits $x \in C_i$, $y \in C_j$ such that $x\ne y$ and $\{x,y\} \not\in E$. Therefore, the claim follows from Proposition \ref{prop_lmcc4}.
\end{proof}

In fact, a much stronger characterization exists for \emph{assignment-minimum clique cover}.

\begin{proposition}
\label{prop_lmcc6}
Let $\C= \{C_1, C_2, \ldots, C_k\}$ be an \emph{edge clique cover} of a graph $G=(V,E)$. If $\C$ is an \emph{assignment-minimum clique cover} of $G$, then there exists a \emph{locally minimal clique cover} $\C^{*}$ of $G$ such that $|\C^{*}| = |\C|$ and $C^{*}_l = C_l$ for all $l \in [k]$.
\end{proposition}
\begin{proof}
Any \emph{assignment-minimum clique cover} satisfies the property that for each pair of distinct cliques $\{C_i, C_j\} \in \C$, there exits $x \in C_i$, $y \in C_j$ such that $x\ne y$ and $\{x,y\} \not\in E$.  Therefore, we can use the construction described in Proposition \ref{prop_lmcc4} to get a clique cover $\C^*$ such that $|\C^*| = |\C|$. For \emph{assignment-minimum clique cover}, it must be the case that $C^{*}_l = C_l$ for all $l \in [k]$; otherwise that would contradict the minimality of number of individual assignments of vertices to cliques in $\C$.
\end{proof}

A characterization that holds for both \emph{minimal clique cover} and \emph{locally minimal clique cover} is as follows.

\begin{proposition}
\label{prop_lmcc7}
Let $\C= \{C_1, C_2, \ldots, C_k\}$ be an \emph{edge clique cover} of a graph $G=(V,E)$. If $\C$ is \emph{locally minimal}, then for each pair of distinct cliques $\{C_i, C_j\} \in \C$, $C_i \not\subseteq C_j$.
\end{proposition}
\begin{proof}
For the sake of contradiction assume $\C$ is \emph{locally minimal} and there exists a pair of distinct cliques $\{C_i, C_j\} \in \C$ such that $C_i \subseteq C_j$. This implies there exists a pair of distinct cliques $\{C_i, C_j\} \in \C$ such that for all $x \in C_i$, and for all $y \in C_j$, $x\ne y$, $\{x,y\} \in E$. By Proposition \ref{prop_lmcc2}, we have a contradiction to the assumption that $\C$ is \emph{locally minimal}.
\end{proof}

Note that Proposition \ref{prop_lmcc2} implies Proposition \ref{prop_lmcc7}, but the converse does not hold. Our next characterization shows that a tight upper bound on the size of \emph{minimum clique cover} (Theorem 2 \cite{erdos1966representation}) also holds for \emph{locally minimal clique cover}.

\begin{proposition}
\label{prop_lmcc8}
Let $\C = \{C_1, C_2, \ldots, C_k\}$ be an \emph{edge clique cover} of a graph $G=(V,E)$ with $n$ vertices. If $\C$ is locally minimal, then $k \leq \lfloor \frac{n^2}{4} \rfloor$.
\end{proposition}
\begin{proof}
We induct on $n$. It is straightforward to see that the claim holds for $n=2$ and $n=3$. Assume that the claim holds for all $n_0 \leq n$.

Let $\C = \{C_1, C_2, \ldots, C_k\}$ be a \emph{locally minimal clique cover} of a graph $G=(V,E)$ with $n+2$ vertices. Let $\{x,y\}$ be an edge of $G$, and $E_{xy}$ be the set of edges incident to $x$ or $y$, including the edge $\{x,y\}$, i.e., $E_{xy} = \{\{w,z\} \in E | z \in \{x,y\}\}$. Let $H=(V_H,E_H)$ be the subgraph induces by vertices of $V\backslash \{x,y\}$ in $G$. Note that $E_H = E \backslash E_{xy}$.

Let $\C^{*} = \{C_1, C_2, \ldots, C_{k^{*}}\}$ be a \emph{locally minimal clique cover} of $H$. Since $|V_H|= n$, by inductive hypothesis, $k^{*} \leq \lfloor \frac{n^2}{4} \rfloor$.

Number of edges in $E_{xy}$ is at most $2n+1$. By the construction of \emph{locally minimal clique cover}, number of additional cliques needed to cover edges of $E_{xy}$ is at most $n+1$. This can be seen as follows.

Let $z$ be any vertex of $V_H$. Assume $\{x,z\}$ and $\{y,z\}$ exist in $G$. WLOG, assume we cover the edge $\{x,z\}$ first, either with a clique of $\C^{*}$ or a new clique $C_{k^{*}+1}= \{x,z\}$. Next, we may be able to cover the edge $\{y,z\}$ with a clique of $\C^{*}$; if not, then the clique $C_{k^{*}+1}= \{x,z\}$, we just have created, is sufficient to cover the edge $\{y,z\}$, since the edge $\{x,y\}$ exists in $G$. In any of the cases, to cover edges incident from $x$ and $y$ to $z$, at most one additional clique would be needed (this also holds if any of the edges $\{x,z\}$ and $\{y,z\}$ does not exist in $G$). Therefore, at most $n+1$ additional cliques would be needed, considering all $n$ vertices of $V_H$ and the edge $\{x,y\}$.

Note that the ordering of covering edges of $E_{xy}$ does not affect the preceding argument. To see this, let $z_1$ and $z_2$ be two distinct vertices of $V_H$. If the edges $\{x,z_1\}$, $\{x,z_2\}$, and $\{z_1,z_2\}$ exist in $G$, then we may cover these edges with a clique $\{x,z_1,z_2\}$. If the edge $\{y,z_1\}$ and $\{y,z_2\}$ both exist in $G$, then we would not need any additional cliques, since the clique $\{x,y,z_1, z_2\}$ would cover all the edges. Now, if only one of $\{y,z_1\}$ and $\{y, z_2\}$ exits in $G$, then we may need one additional clique. Therefore, to cover edges incident on $z_1$ and $z_2$ from $x$ and $y$, at most two additional cliques would be needed. Generalizing this argument, we see that any $\{z_1, z_2,\ldots, z_q\} \subseteq V_H$, we would need at most $q$ additional cliques to cover all edges incident from $x$ and $y$ to vertices of $\{z_1, z_2,\ldots, z_q\}$.

We conclude that $k \leq k^{*} + n + 1 \leq \lfloor \frac{n^2}{4} \rfloor + n +1 = \lfloor \frac{(n+2)^2}{4} \rfloor$.
\end{proof}

Proposition \ref{prop_lmcc8} gives us a characterization of compactness of \emph{locally minimal clique cover} in terms of number of cliques. In subsequent characterizations, we show compactness of \emph{locally minimal clique cover} in terms of vertex assignments and edge assignments.

\begin{proposition}
\label{prop_lmcc9}
Let $\C = \{C_1, C_2, \ldots, C_k\}$ be an \emph{edge clique cover} of a graph $G=(V,E)$, and $N(x)$ denote the neighbours of a vertex $x$. If $\C$ is \emph{locally minimal}, then each vertex $x$ appears in at most $|N(x)|$ distinct cliques of $\C$.
\end{proposition}
\begin{proof}
Let $\C^{*}$ be a \emph{locally minimal clique cover} of a subgraph of $G$, and let $\{x,y\}$ be an edge not included in the cliques of $\C^{*}$. Now, consider the expansion of $\C^{*}$ for $\{x,y\}$. Before the expansion, none of the cliques in $\C^{*}$ contained both $x$ and $y$. After the expansion, $\{x,y\}$ would not get selected for expansion in subsequent choices. The number of uncovered edges $\{x,y\}$ incident on a vertex $x$ is at most $|N(x)|$. Therefore, the number of such expansions of clique covers for a vertex $x$ is at most $|N(x)|$, i.e., $x$ appears at most $|N(x)|$ times in $\C$.
\end{proof}

\begin{proposition}
\label{prop_lmcc10}
Let $\C = \{C_1, C_2, \ldots, C_k\}$ be an \emph{edge clique cover} of a graph $G=(V,E)$ with $m$ edges. If $\C$ is \emph{locally minimal}, then the number of individual assignments of vertices to cliques of $\C$ is at most $2m$.
\end{proposition}
\begin{proof}
Since $\C$ is \emph{locally minimal}, by Proposition \ref{prop_lmcc9}, each vertex $x$ appears in at most $|N(x)|$ distinct cliques of $\C$. Therefore, $\sum_{l \in [k]} |C_l| \leq \sum_{x\in V} |N(x)| = 2m$.
\end{proof}

Next characterization follows from Proposition \ref{prop_lmcc10} and Proposition \ref{prop2}.

\begin{proposition}
Let $\C = \{C_1, C_2, \ldots, C_k\}$ be an \emph{edge clique cover} of a graph $G=(V,E)$ with \emph{degeneracy} $d$. If $\C$ is \emph{locally minimal}, then on average a vertex is included in at most $2d$ cliques of $\C$, i.e., $ \frac{\sum_{C_l \in \C}|C_l|}{n} < 2d$.
\end{proposition}

\begin{proposition}
\label{prop_lmcc11}
Let $\C = \{C_1, C_2, \ldots, C_k\}$ be an \emph{edge clique cover} of a graph $G=(V,E)$ with $m$ edges and $\Delta$ maximum degree. If $\C$ is \emph{locally minimal}, then the number of individual assignments of edges to cliques of $\C$ is at most $m\Delta$.
\end{proposition}
\begin{proof}
Consider any edge $\{x,y \} \in E$. Since $\C$ is \emph{locally minimal}, by Proposition \ref{prop_lmcc9}, any vertex $x\in V$ appears in at most $|N(x)|$ distinct cliques of $\C$. Both $x$ and $y$ can appear in at most $\min \{|N(x)|, |N(y)|\} \leq \Delta$ distinct cliques of $\C$. Therefore, $\sum_{\{x,y\} \in E} |\{C_l | \{x,y\} \in C_l, C_l \in \C\}| \leq m \Delta$.
\end{proof}

\subsection{Implicit Set Representation}

From the definition of \emph{locally minimal clique cover}, a natural question emerges: from an \emph{edge clique cover} of a subgraph how would one quickly find a clique that can cover an edge $\{x,y\}$ or report that none exists. Obviously, scanning the entire set of cliques of the clique cover is an inefficient method. We address this problem by introducing a set representation of graph connected to intersection graph theory \cite{mckee1999topics}.

\begin{definition}[Intersection Graph]
\label{def_ig}
Let $\F = \{F_1, F_2, \ldots, F_n\}$ be a family of sets. The \emph{intersection graph} of $\F$ is a graph that has $\F$ as the vertex set, and an edge for each pair of distinct sets $F_x$ and $F_y$ if and only if $F_x \cap F_y \ne \emptyset$.
\end{definition}

On the other hand, every graph is an intersection graph of some family of sets \cite{erdos1966representation}, which leads us to the following definition.

\begin{definition}[Set Representation]
Let $G=(V,E)$ be a graph with $n$ vertices. A family of sets $\F = \{F_1, F_2, \ldots, F_n\}$ is called a \emph{set representation} of $G$ if $\{x,y\} \in E$ if and only if $F_x \cap F_y \ne \emptyset$.
\end{definition}

Note that \emph{intersection graph} of a family of sets is unique, but a graph can have many set representations.

Let $U_{\F} = \cup_{x\in [n]} F_x$. Every element $l \in U_{\F}$ corresponds to a clique of $G$ such that the set $\{F_x | l \in F_x\}$ corresponds to the vertex set of the clique. Each set $F_x$ corresponds to a set of cliques of $G$ that contain the vertex $x$. By definition of \emph{set representation}, for every edge $\{x,y\} \in E$, there exists $l \in U_{\F}$ such that $l \in F_x \cap F_y$. Therefore, the set $U_{\F}$ corresponds to an \emph{edge clique cover} of $G$.

\begin{definition}[Intersection Graph Basis \cite{garey1979computers}]
\label{def_igb}
Let $G=(V,E)$ be a graph with $n$ vertices and $\F=\{F_1, F_2, \ldots, F_n\}$ be a \emph{set representation} of $G$. Let $U_{\F}=\cup_{x\in [n]} F_x$. If $|U_{\F}|$ is minimum over all \emph{set representations} of $G$, then $U_{\F}$ is called an \emph{intersection graph basis} of $G$.
\end{definition}

As a consequence, computing a \emph{minimum clique cover} of $G$ is equivalent to computing an \emph{intersection graph basis} of $G$ (see Theorem 1.6 of \cite{mckee1999topics} for a proof). Although every \emph{set representation} of a graph is in one-to-one correspondence with an \emph{edge clique cover} of the graph, that does not give us an efficient way to construct either of them from scratch. Next, we introduce a constructive way of representing graph that implicitly contains a \emph{set representation}.

\begin{definition}[Representative Set]
\label{def_rs}
Let $\C = \{C_1, C_2, \ldots, C_k\}$ be an \emph{edge clique cover} of a subgraph of $G=(V,E)$. We call a clique $C_l \in \C$ a \emph{representative} of a vertex $x \in V$ if either (1) $x \in C_l$ or (2) $x \not\in C_l$ and $C_l \subseteq N(x)$. A \emph{representative set} $R_x$ of a vertex $x$ is the set of (indices of) representatives of $x$ in $\C$, i.e., $R_x= \{l \in [k] | x \in C_l \text{ or } (x \not\in C_l \text{ and } C_l \subseteq N(x))\}$.
\end{definition}

\begin{definition}[Implicit Set Representation]
\label{def_isr}
Let $\C$ be an \emph{edge clique cover} of graph $G=(V,E)$ and $R_x$ be the corresponding \emph{representative set} for a vertex $x$. We call the family of sets $\R = \{R_1, R_2, \ldots , R_n\}$ an \emph{implicit set representation} of $G$.
\end{definition}

\begin{figure}
    \centering
    \resizebox{7.5cm}{5.0cm}{%
         \tikzset{main node/.style={text=blue, circle,fill=black!20,draw,minimum size=0.8cm,inner sep=0pt},
            }
 \begin{tikzpicture}
    \node[main node] (1) {$x$};
    \node[main node] (2) [above right = 1.0cm and 1.0cm of 1] {$w$};
    \node[main node] (3) [below = 2.0cm of 2] {$z$};
    \node[main node] (4) [right = 2.5cm of 1] {$y$};
    
    \node[main node] (5) [left = 2.0cm of 3] {$a$};
    
    \node[main node] (6) [left = 2.0cm of 2] {$b$};
    
    \node[main node] (7) [right = 2.5cm of 2] {$c$};
    
    \node[main node] (8) [right = 2.5cm of 3] {$d$};

    \path[draw,thick]
    (1) edge node {} (2)
    (1) edge node {} (3)
    (2) edge node {} (3)
    (2) edge node {} (4)
    (3) edge node {} (4)
    
    (5) edge node {} (1)
    (5) edge node {} (3)
    
    (6) edge node {} (1)
    (6) edge node {} (2)
    
    (7) edge node {} (4)
    (7) edge node {} (2)
    
    (8) edge node {} (4)
    (8) edge node {} (3)
    
    ;
\end{tikzpicture} 
    }
    \caption{An example graph $G$ showing $R_x \cap R_y \ne \emptyset$ does not imply $\{x,y\} \in E$.
    A \emph{minimum clique cover} $\C$ of $G$ consists of following cliques:
    \textcolor{blue}{$C_1 = \{x,a,z\}$},
    \textcolor{blue}{$C_2 = \{x,b,w\}$},
    \textcolor{blue}{$C_3 = \{y,c,w\}$},
    \textcolor{blue}{$C_4 = \{y,d,z\}$},
    \textcolor{blue}{$C_5 = \{w,z\}$}.
    For all $v \in \{a,b,c,d,z,w\}$, \emph{set representation} and \emph{implicit set representation} of $G$ have the same sets for $F_v$ and $R_v$ respectively.
    But, \emph{set representation} has 
    \textcolor{blue}{$F_x = \{1,2\}$}, 
    \textcolor{blue}{$F_y = \{3,4\}$}, 
    whereas \emph{implicit set representation} has 
    \textcolor{blue}{$R_x = \{1,2,5\}$}, 
    \textcolor{blue}{$R_y = \{3,4,5\}$}. }
    \label{fig:isr}
\end{figure}
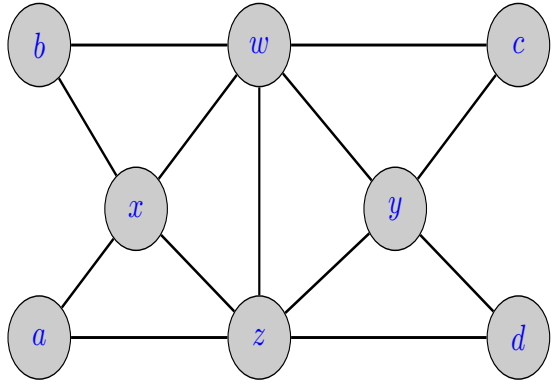

Note that \emph{set representation} requires biconditional for every edge: $\{x,y\} \in E \Longleftrightarrow F_x \cap F_y \ne \emptyset$. \emph{Implicit set representation} can only ensure one way implication: $\{x,y\} \in E \implies R_x \cap R_y \ne \emptyset$. Figure \ref{fig:isr} shows an example where $R_x \cap R_y \ne \emptyset$ does not imply $\{x,y\} \in E$. It would be evident from following characterization that an \emph{implicit set representation} of a graph always contains a \emph{set representation} of the graph. 

\begin{proposition}
\label{prop_isr}
Let $\C=\{C_1, C_2,\ldots, C_k\}$ be an \emph{edge clique cover} of a graph $G$ with $n$ vertices, and $\R$ be the corresponding \emph{implicit set representation} of $G$. Let $U_{\R} = \cup_{x \in V} R_x$. If $|U_{\R}|$ is minimum over all \emph{implicit set representations} of $G$, then $\C$ is a \emph{minimum clique cover} of $G$.
\end{proposition}
\begin{proof}
Let $\F=\{F_1, F_2,\ldots, F_n\}$ be a family of sets where $F_x = R_x \backslash \{l \in [k] | x \not\in C_l\}$. Let $U_{\F}= \cup_{x\in V} F_x$. Note that $\F$ is a \emph{set representation} of $G$, since $F_x = \{l | x \in C_l\}$.
Moreover, $U_{\F} = U_{\R}$. $|U_{\F}|$ is minimum over all \emph{set representation} of $G$. To see this, suppose $|U_{\F}|$ is not minimum, i.e., there exists $l \in U_{\F}$ such that $C_l$ can be removed from $\C$, which contradicts the assumption that $|U_{\R}|$ is minimum. Therefore, $U_{\F}$ is an \emph{intersection graph basis} of $G$. Since $U_{\F}$ is an \emph{intersection graph basis} of $G$, $\C$ is a \emph{minimum clique cover} of $G$.
\end{proof}

As in Proposition \ref{prop_isr}, we can obtain similar characterizations for \emph{assignment-minimum clique cover}.

\begin{proposition}
Let $\C=\{C_1, C_2,\ldots, C_k\}$ be an \emph{edge clique cover} of a graph $G$ with $n$ vertices, and $\R$ be the corresponding \emph{implicit set representation} of $G$. If $\sum_{x \in V} |R_x \backslash \{l \in [k] | x \not\in C_l\} |$ is minimum over all \emph{implicit set representations} of $G$, then $\C$ is an \emph{assignment-minimum clique cover} of $G$.
\end{proposition}

Definition \ref{def_rs} captures the idea that a \emph{representative} of a vertex $x$ corresponds to a clique that may be used to cover an edge incident on $x$, regardless of whether $x$ is contained in the clique. So far it may not be evident how Definition \ref{def_rs} resolves the problem of efficiently constructing a \emph{set representation} from scratch. In following characterization we provide a clarification.

\begin{proposition}
\label{prop_isr2}
Let $\C= \{C_1, C_2, \ldots, C_k\}$ be an \emph{edge clique cover} of a subgraph of $G=(V,E)$, and $R_x$ be the corresponding \emph{representative set} of a vertex $x$. Let $\{x,y\} \in E$ be an uncovered edge. There exists representative set $l \in R_x \cap R_y$ if and only if $\{x,y\}$ can be covered by a clique $C_l \in \C$.
\end{proposition}
\begin{proof}
\textbf{Only if.} Since $\{x,y\}$ is an uncovered edge, both $x$ and $y$ are not contained in any of the cliques of $\C$. By the definition of \emph{representative set}, one of the following holds for $C_l$: (i) $x \in C_l$, $y \not \in C_l, C_l \subseteq N(y)$, (ii) $y \in C_l$, $x \not \in C_l, C_l \subseteq N(x)$, (iii) $x \not\in C_l$, $y \not \in C_l, C_l \subseteq N(x), C_l \subseteq N(y)$. Each of the three cases, both $x$ and $y$ can be included in $C_l$ and thus $\{x,y\}$ can be covered by $C_l$.

\textbf{If.} Since the edge $\{x,y\}$ is uncovered, if $x$ is contained in $C_l$, then it must be the case $y \not\in C_l$ and $C_l \subseteq N(y)$. It follows that if $x$ is contained in $C_l$, then $l \in R_x$ and $l \in R_y$. Similarly, $l \in R_x$ and $l \in R_y$ if $y$ is contained in $C_l$ or both $x$ and $y$ are not contained in $C_l$. Therefore, $l \in R_x \cap R_y$.
\end{proof}

Proposition \ref{prop_isr2} immediately connects us to the construction of \emph{locally minimal clique cover}. To cover an edge $\{x,y\}$, we can simply select a clique $C_l \in \C$ if $l$ is contained in $R_x \cap R_y$. Each such choice to cover an edge may require update of some sets in $\R$ to conform to Definition \ref{def_rs}. We describe the updates in the context of algorithm design (Section \ref{sec_algos}). Our next characterization provides a bound on the space requirement of \emph{implicit set representation} for constructing a \emph{locally minimal clique cover}.

\begin{proposition}
\label{prop_isr3}
Let $\C= \{C_1, C_2, \ldots, C_k\}$ be an \emph{edge clique cover} of a graph $G=(V,E)$ with $m$ edges, and $\R$ be the corresponding \emph{implicit set representation} of $G$. If $\C$ is \emph{locally minimal}, then $\sum_{x\in V} |R_x| < 2m+k\Delta$.
\end{proposition}
\begin{proof}
From the definition of \emph{representative set}, we have
$$\sum_{x\in V} |R_x|= \sum_{x\in V} |\{l | x \in C_l\}| + \sum_{x \in V} |\{l | x\not\in C_l, C_l \subseteq N(x) \}|.$$

For the first sum on the right hand side, by Proposition \ref{prop_lmcc10}, we have
$$\sum_{x\in V} |\{l | x \in C_l\}| = \sum_{C_l \in \C} |C_l| \leq 2m.$$

It remains to show a bound for the second sum of the right hand side.

Since $\C$ is \emph{locally minimal}, there exists a permutation $\pi$ of $\{1,\ldots,k\}$ such that $C_{\pi(i)}$ contains an edge $\{x,y\}$ not contained in any $C_{\pi(j)}, j <i$. For the permutation $\pi$, existence of $C_{\pi(i)}$ is necessary for covering $\{x,y\}$. WLOG we can assume that in the \emph{locally minimal} construction of $\C$, the edge $\{x,y\}$ triggered the creation of $i^{th}$ clique in $\C$. Immediately after including $\{x,y\}$ in the $i^{th}$ clique of $\C$, we would have $\pi(i) \in R_z$ for all $z \in N[x] \cap N[y]$ (from the definition of \emph{representative set}).

For any subsequent edge $\{x,y\}$ added to the $i^{th}$ clique of $\C$ (for locally minimal construction), $\pi(i) \in R_x \cap R_y$ holds by Proposition \ref{prop_isr2}. Therefore, all \emph{representative sets} that include $\pi(i)$ have already been accounted for, when we created (by \emph{locally minimal} construction) the $i^{th}$ clique of $\C$.

Any vertex $z$ such that $z \not\in C_{\pi(i)}$, but $C_{\pi(i)} \subseteq N(z)$, must be a common neighbour of $x$ and $y$ such that the edge $\{x,y\}$ required the creation of $C_{\pi(i)}$. And $|\{z | z \not\in C_{\pi(i)}, C_{\pi(i)} \subseteq N(z)\}| \leq |N(x) \cap N(y)| < \Delta$. Summing over all cliques we get, $$\sum_{C_l \in \C} |\{x | x \not\in C_l, C_l \subseteq N(x)\}| < k\Delta.$$

Consider a bipartite graph $B=(V, \C, E_B)$ where $E_B$ is the relation of vertices and cliques from the second condition of \emph{representative}: $E_B = \{\{x, C_l\}| x\in V, C_l \in \C, x \not\in C_l, C_l \subseteq N(x) \}$. Sum of the degrees of the vertices on the $V$ side of $B$ is equal to the sum of the degrees of the vertices on the $\C$ side of $B$. Therefore, $$\sum_{x \in V} |\{l | x\not\in C_l, C_l \subseteq N(x) \}| = \sum_{C_l \in \C} |\{x | x \not\in C_l, C_l \subseteq N(x)\}| < k\Delta.$$
\end{proof}

\begin{remark}
A tighter bound than Proposition \ref{prop_isr3} is possible if we construct \emph{locally minimal clique cover} by choosing the edges in the order of \emph{DEP} of $G$, and keep maintaining only necessary \emph{representative sets}. At any stage of the construction of \emph{locally minimal clique cover}, let $\{x,y\}$ be the first uncovered edge of \emph{DEP} where $y \in N_d(x)$. Note that if $x=u_i$, then all edges incident on $u_j$ such that $j < i$ are covered. Therefore, if we choose to add a new clique for $\{x,y\}$, then we only need to keep this information for the \emph{representative sets} of $N_d[x] \cap N[y]$. As a consequence, the $k\Delta$ bound in the proof of Proposition \ref{prop_isr3} can be replaced with $dk$, since $|N_d(x) \cap N(y)| < d$. Overall, we would have $O(m+dk)$ space bound. Furthermore, since we can assume $d= O(\sqrt{m})$, for $k=O(\sqrt{m})$ this bound becomes $O(m)$.
\end{remark}

\section{New Framework: Algorithms}
\label{sec_algos}

In this section, we describe a new set of bounded search tree algorithms for the problems listed in Section \ref{intro_problems}. The algorithms are based on the concepts introduced in Section \ref{sec_bblock}. The general theme of these search tree algorithms is as follows. At each node of a search tree, the corresponding algorithm branches on a set of cliques already constructed for a partial solution (clique cover of a subgraph). Consequently, branching factors of a search tree can be bounded in terms of number of cliques in a clique cover or other parameters. Most notably, we would be able to show that the size of a search tree is structurally dependent on the \emph{clique number} of graph, a feature lacking in the algorithms we have presented in Section \ref{sec_ecrs}.

\subsection{Edge Clique Cover}
\label{algos_ecc}

We will continue to assume that data reduction rules are applied on an instance $(G,k)$ of \emph{ECC}, including the rules described in Section \ref{ecrs_ecc}. Figure \ref{fig:ecc_algo2} shows our new bounded search tree algorithm for \emph{ECC}, henceforth referred as \emph{ECCS2}. For better readability, we describe the details of steps 3a, 3d, 4a, and 4d separately as subroutines in Figure \ref{fig:ecc_algo2_subs}.

\emph{ECCS2} works as follows. At every node of the search tree, it selects the \emph{last} uncovered edge $\{x,y\}$ from the \emph{DEP} of $G$ such that $y \in N_d(x)$. Then, it branches on each of the choices of \emph{representatives} $l \in R_x \cap R_y$, by covering the edge $\{x,y\}$ with the clique $C_l$ (steps 3a-3d). If all of these branches fail and the parameter $k$ permits, then it creates an additional branch (steps 4a-4d), by covering the edge $\{x,y\}$ with a new clique.

Algorithm \emph{ECCS2} works by modifying a single \emph{edge clique cover} $\C$ and a single \emph{implicit set representation} $\R$ at every node of a search tree. This makes the algorithm space efficient. Now, to conform to the definition of \emph{representative set}, we need to update $\R$ in response to changes in the cliques of $\C$ (steps 3a, 3d), or inclusion/exclusion of clique in $\C$ (steps 4a, 4d). Next, we expand on these updates, included in the subroutines of \emph{ECCS2} in Figure \ref{fig:ecc_algo2_subs}.

We use a data structure $\D$ to help the updates of sets in $\R$. A set $D_l \in \D$ contains a set of vertices that have clique $C_l$ in their \emph{representative sets}, i.e., $D_l = \{x \in V | l \in R_x, C_l \in \C\}$. The sets in $\R$ contain a mapping of vertices to the representatives (cliques), whereas the sets in $\D$ contain the corresponding inverse mapping of the representatives (cliques) to the vertices.

In the subroutine for step 4a, a new clique, with index $q$, is included in the solution. By the definition of \emph{representative set}, this prompts $q$ to be included in the \emph{representative sets} of all vertices in $N[x] \cap N[y]$. Similarly, in the subroutine for step 4d, the clique $q$ is excluded from the clique cover, which prompts $q$ to be excluded from the representative sets of all vertices in $N[x] \cap N[y]$.

In the subroutine for step 3a, since the edge $\{x,y\}$ to be included in $C_l$, any vertex $z$ such that $l \in R_z$ must have both $x$ and $y$ in the neighbourhood of $z$. The set of vertices that violates this requirement is collected in a set $U$ (a local data structure allocated for every node in a search tree), and corresponding \emph{representative sets} of these vertices are updated. Indicator variable $x_l$ (resp. $y_l$) is set that would denote whether the vertex $x$ (resp. $y$) is already contained in $C_l$. The subroutine for step 3d simply undoes the updates of $\R$ and $\D$ using the set $U$. The indicator variables are used to restore the clique $C_l$ to the state prior to including the edge $\{x,y\}$ in $C_l$. Each of the subroutines of Figure \ref{fig:ecc_algo2_subs} takes $O(\Delta)$ time.

\begin{figure}
    \centering
    \fbox{\parbox{5.5in}{
\underline{Edge-Clique-Cover-Search2$(G, k, \C, \R)$:}

\textcolor{teal}{\\ // Abbreviated \textbf{ECCS2}$(G,k,\C,\R)$}

    \begin{enumerate}
        \item if $\C$ covers edges of $G$, then return $\C$
        \item select the \emph{last} uncovered edge $\{x,y\}$ from the \emph{DEP} of $G$ where $y \in N_d(x)$
        
        \textcolor{teal}{// cf. step 3 of \emph{ECCS} in Figure \ref{fig:ecc_algo1}}
        
        \item for each $l \in R_x \cap R_y$ do 
    \begin{enumerate}
        \item \textcolor{blue}{cover the edge $\{x,y\}$ with the clique $C_l$ and update $\R$}  \textcolor{teal}{// Figure \ref{fig:ecc_algo2_subs} (a)}
        \item $\Q \leftarrow \textbf{ECCS2}(G, k, \C, \R)$
        \item if $\Q \ne \emptyset$, then return $\Q$
        \item \textcolor{blue}{undo changes done to $\C$ and $\R$ at step 3a} \textcolor{teal}{// Figure \ref{fig:ecc_algo2_subs} (b)}
    \end{enumerate}
    \item if $k > 0$, then
    \begin{enumerate}
        \item \textcolor{blue}{set $\C$ to $\C \cup \{\{x,y\}\}$ and update $\R$} \textcolor{teal}{// Figure \ref{fig:ecc_algo2_subs} (c)}
        \item $\Q \leftarrow \textbf{ECCS2}(G, k-1, \C, \R)$
        \item if $\Q \ne \emptyset$, then return $\Q$
        \item \textcolor{blue}{undo changes done to $\C$ and $\R$ at step 4a} \textcolor{teal}{// Figure \ref{fig:ecc_algo2_subs} (d)}
    \end{enumerate}
    \item return $\emptyset$
    \end{enumerate} 
}}   
    \caption{A bounded search tree algorithm for \emph{ECC}, denoted \emph{ECCS2}.}
    \label{fig:ecc_algo2}
\end{figure}

\begin{figure}
    \centering
    \begin{minipage}{0.45\textwidth}
    \begin{subfigure}{\linewidth}
        \fbox{\parbox{\linewidth}{
        \textcolor{blue}{\underline{Step 3a:}}
        
        \begin{enumerate}
            \item $U \leftarrow \emptyset$
            \item for each $z \in D_l \backslash \{x,y\}$ such that $\{x,z\} \not\in E$ or $\{y,z\} \not\in E$ do
            
            \begin{enumerate}
                \item $U \leftarrow U \cup \{z\}$
                \item $R_z \leftarrow R_z \backslash \{l\}$
            \end{enumerate}
        
            \item $D_l \leftarrow D_l \backslash U$
            \item $x_l \leftarrow \mathbbm{1}_{[x \in C_l]}$
            \item $y_l \leftarrow \mathbbm{1}_{[y \in C_l]}$
            \item $C_l \leftarrow C_l \cup \{x,y\}$
        \end{enumerate}
        }}  
    \caption{}
    \end{subfigure}
    
    \begin{subfigure}{\linewidth}
        \fbox{\parbox{\linewidth}{
        \textcolor{blue}{\underline{Step 3d:}}
        
        \begin{enumerate}
            \item for each $z \in U$ do
            \begin{enumerate}
                \item $R_z \leftarrow R_z \cup \{l\}$
                \item $D_l \leftarrow D_l \cup \{z\}$
            \end{enumerate}
            \item if $x_l =0 $, then set $C_l$ to $C_l \backslash \{x\}$
            \item if $y_l =0 $, then set $C_l$ to $C_l \backslash \{y\}$
        \end{enumerate}
        }}  
    \caption{}
    \end{subfigure}
    \end{minipage}
    \hspace{\fill}
    \begin{minipage}{0.45\textwidth}
         \begin{subfigure}{\linewidth}
        \fbox{\parbox{\linewidth}{
        \textcolor{blue}{\underline{Step 4a:}}
        
        \begin{enumerate}
            \item $q \leftarrow |\C| + 1$
            \item $C_q \leftarrow \{x,y\}$
            \item $D_q \leftarrow N[x] \cap N[y]$
            \item for each $z \in D_q$ do
            \begin{enumerate}
                \item $R_z \leftarrow R_z \cup \{q\}$
            \end{enumerate}
        \end{enumerate}
        }}  
    \caption{}
    \end{subfigure}
    
     \begin{subfigure}{\linewidth}
        \fbox{\parbox{\linewidth}{
        \textcolor{blue}{\underline{Step 4d:}}
        
        \begin{enumerate}
            \item $q \leftarrow |\C|$
            \item $\C \leftarrow \C \backslash \{C_q\}$
            \item for each $z \in D_q$ do
            \begin{enumerate}
                \item $R_z \leftarrow R_z \backslash \{q\}$
            \end{enumerate}
        \end{enumerate}
        }}  
    \caption{}
    \end{subfigure}
    \end{minipage}
    \caption{Subroutines of \emph{ECCS2} shown in Figure \ref{fig:ecc_algo2}.}
    \label{fig:ecc_algo2_subs}
\end{figure}

The definition of \emph{locally minimal} construction is oblivious of edge permutation. An exact algorithm that strictly adheres to the definition of \emph{locally minimal} construction would potentially need to try all permutations of edges (for \emph{ECCS}, this would be equivalent to trying all possible choices of uncovered edge at step 2). Consequently, we would have algorithms with large dependency on parameters (for example, we would need $2^{O(\beta k^2)}n^{O(1)}$ or $2^{\beta k \log (\beta^2 k^3)}n^{O(1)}$ running time for \emph{ECC}). Furthermore, we want to impose permutations of edges (such as step 2 of \emph{ECCS2}) so that we can bound the branching factors of a search tree with certain choices of parameters (such as degeneracy). We address these issues by judicious relaxation of \emph{locally minimal} construction.

If the edge $\{x,y\}$ (selected at step 2) can be covered by a clique $C_l \in \C$, then \emph{ECCS2} selects exactly one such clique from the set $R_x \cap R_y$. Therefore, every branch of step 3 conforms to the \emph{locally minimal} construction. If all branches of step 3 fail, only then the edge $\{x,y\}$ may need to be covered with a new clique in step 4a. If $R_x \cap R_y = \emptyset$, then step 4 still conforms to the \emph{locally minimal} construction. If $R_x \cap R_y \ne \emptyset$ and all branches of step 3 fail, then \emph{ECCS2} relaxes the requirement of \emph{locally minimal} construction. In this case, \emph{ECCS2} considers the possibility of a branch with at most $k$ cliques where the edge $\{x,y\}$ is covered with a clique, not contained in the current partial solution $\C$. Therefore, if $k$ permits, then \emph{ECCS2} includes a new clique in the solution. Incidentally, this type of adaptation would be required for other objective of clique covers, where we would need to include an edge to a specified number of cliques.

Preceding adaptation may seem puzzling, considering the fact that with a simple relaxation of \emph{locally minimal} construction we are able to avoid enumerating all possible choices of an uncovered edge at step 2 of \emph{ECCS2}. We shed some light on this. 

Let $(G,k)$ be a \emph{YES} instance of $\emph{ECC}$, and $\C^{*}$ be a corresponding clique cover. Let us call a clique cover $\C$ of a subgraph of $G$ a \emph{subcovering} of $\C^{*}$ if every clique $C_i \in \C$ is contained in a clique $C^{*}_l \in \C^{*}$, and no pair of distinct cliques $C_i, C_j \in \C$ are contained in the same clique $C^{*}_l \in \C^{*} $. In the search tree $T$ of $\emph{ECCS2}$, for $(G,k)$, consider a clique cover $\C$ of a subgraph of $G$, constructed at a node $u^T$. Observe that if the clique cover $\C$ at node $u^T$ is a \emph{subcovering} of $\C^{*}$, then \emph{ECCS2} would not backtrack from node $u^T$, i.e., a clique cover of $G$ would be returned from node $u^T$. To get a \emph{subcovering} $\C$ of $\C^{*}$ at node $u^T$, it is irrelevant in what order we constructed the cliques in $\C$, and in what order edges of individual clique $C_l \in \C$ are included in $C_l$. All that matters is being able to construct \emph{subcoverings} of $\C^{*}$ with more and more cliques. The relaxation at step 4 of \emph{ECCS2} is precisely for this purpose: even if $R_x \cap R_y$ is nonempty, we allow the edge $\{x,y\}$ to be included in a new clique, as a potential extension for a larger subcovering of some $\C^{*}$.

A correctness proof for \emph{ECCS2} can be obtained using induction (see Appendix \ref{sec:app_algos}). Thus we conclude the following.

\begin{restatable}{lemma}{LemmaECCSTwoOK}
\label{lemma_eccs2_ok}
Algorithm \emph{ECCS2} correctly solves the parameterized problem \emph{ECC}.
\end{restatable}

The space use of \emph{ECCS2} can be bounded as follows (see Appendix \ref{sec:app_algos} for a proof).

\begin{restatable}{lemma}{LemmaECCSTwoSpace}
\label{lemma_eccs2_space}
In a search tree with at most $k$ cliques, \emph{ECCS2} takes $O(m+k\Delta)$ space.
\end{restatable}

Next, we describe bounds on the depth and the branching factors of a search tree of \emph{ECCS2}.

\begin{lemma}
\label{lemma_eccs2_depth}
The depth of a search tree of \emph{ECCS2} is at most $\beta k$.
\end{lemma}
\begin{proof}
Every new clique starts with two vertices (step 4a), and each expansion of a clique (step 3a) with an uncovered edge $\{x,y\}$ will include at least one new vertex to the clique. Therefore, number of total possible expansions of any clique is at most $\beta-2$. With at most $k$ cliques allowed in any branches, the depth of a search tree is at most $\beta k$.
\end{proof}

Recall by $\langle u_1, u_2, \ldots, u_n \rangle$ we denote degeneracy ordering of $V$. Since for any $x = u_i$, we only consider $y \in N_d(x)$, any vertex $z \in N(x) \backslash N_d(x)$ would not be included in any clique $C_l \in \C$ until we process a node such that $x = z$. Therefore, by our choice of uncovered edge $\{x,y\}$ for each node of the search tree of \emph{ECCS2}, the following invariant holds. 

\begin{proposition}
\label{inv_eccs2}
At any node of a search tree of \emph{ECCS}, if $x=u_i$, then for all $ C_l \in \C, C_l \subseteq \{u_i, u_{i+1}, \ldots, u_n\}$.
\end{proposition}

\begin{lemma}
\label{lemma_eccs2_bf}
The number of branches at any node of a search tree of \emph{ECCS2} is at most $\min\{k, \binom{d+1}{2}\}$.
\end{lemma}
\begin{proof}
Consider an uncovered edge $\{x,y\}$ at step 2 of \emph{ECCS2}. For the branches of step 3, by Proposition \ref{inv_eccs2}, we only need to consider cliques created for the subgraph induced by vertices of $\{u_i, u_{i+1}, \ldots, u_n\}$ in $G$.

Since  $|N_d[x]| \leq d+1$, maximum number edges in the subgraph induced by the vertices of $N_d[x]$ is at most $\binom{d+1}{2}$. Since $y \in N_d(x)$ and the edge $\{x,y\}$ is uncovered, number of cliques in $\C$ that are also included in the representative sets of $x$ and $y$ is less than $\binom{d+1}{2}$, i.e., $|R_x \cap R_y| < \binom{d+1}{2}$. Considering the branch for a new clique at step 4, number of branches at any node of a search tree of \emph{ECCS2} is at most $\binom{d+1}{2}$.

At any node of the search tree at most $k$ cliques are allowed in $\C$, i.e., $|R_x \cap R_y| \leq k$. If at any node of the search tree $|R_x \cap R_y| = k$, then there is no branch at step 4 for the corresponding node. The claim follows.
\end{proof}

\begin{lemma}
\label{lemma_eccs2_nn}
The number of nodes in a search tree of \emph{ECCS2} is at most $2^{\beta k \log k}$.
\end{lemma}
\begin{proof}
From Lemma \ref{lemma_eccs2_depth}, the depth of a search tree of \emph{ECCS2} is at most $\beta k$. From Lemma \ref{lemma_eccs2_bf}, the branching factors of a search tree of \emph{ECCS2} are bounded by $\min\{k, \binom{d+1}{2}\} \leq k$. Therefore, the number of nodes in a search tree of \emph{ECCS} is at most $k^{\beta k} = 2^{\beta k \log k}$.
\end{proof}

Considering time needed for data reduction, Theorem \ref{thm_ecc2} follows from Lemma \ref{lemma_eccs2_nn} and Proposition \ref{prop_eqv1}.

We have shown in Section \ref{sec_ecrs} that the algorithm \emph{ECCS} has consistently better running time than \emph{ECCG} \cite{gramm2009data}. For a large class of graphs, \emph{ECCS2} further beats \emph{ECCS} in the running time. This can be seen as follows. From Lemma \ref{lemma_eccs2_bf}, we have $\min \{k, \binom{d+1}{2}\} \leq \binom{d+1}{2} < (d+1)^2$. Therefore, the number of nodes in a search tree of \emph{ECCS2} is at most $2^{O(\beta k \log d)}$. For $\beta = o(d / \log d)$, \emph{ECCS2} improves the bound on search tree size by a factor of $\frac{2^{O(dk)}}{2^{O(\beta k \log d)}} = 2^{O(k(d -\beta \log d))} = 2^{O(dk)}$, an exponential improvement over the bound on search tree size. For many instances of \emph{ECC}, $\beta = o(d / \log d)$ is a very mild requirement, considering that $d$ can grow linearly with the input size, while $\beta$ remains constant (Remark \ref{rem_beta}).

\subsection{Assignment Clique Cover}
\label{algos_acc}

We will continue to assume that data reduction rules are applied on an instance $(G,t)$ of \emph{ACC}, including the rules described in Section \ref{ecrs_acc}. Figure \ref{fig:acc_algo2} shows our new bounded search tree algorithm for \emph{ACC}. henceforth referred as \emph{ACCS2}. Steps 4c, 4f, 5a, 5d of \emph{ACCS2} are identical to the steps 3a, 3d, 4a, 4d of \emph{ECCS2} described in Figure \ref{fig:ecc_algo2_subs}.

\begin{figure}
    \centering
    \fbox{\parbox{5.5in}{
\underline{Assignment-Clique-Cover-Search2$(G, t, \C, \R)$:}

\textcolor{teal}{\\ // Abbreviated \textbf{ACCS2}$(G,t,\C,\R)$}

    \begin{enumerate}
        \item if $\C$ covers edges of $G$, then return $\C$
        \item if $t \leq 0$ return $\emptyset$
        \item select the \emph{last} uncovered edge $\{x,y\}$ from the \emph{DEP} of $G$ where $y \in N_d(x)$
        
        \textcolor{teal}{// cf. step 3 of \emph{ACCS} in Figure \ref{fig:acc_algo1}}
        \item for each $l \in R_x \cap R_y$ do 
    \begin{enumerate}
        \item $s \leftarrow |\{x,y\} \backslash (C_l \cap \{x,y\})|$
        \item if $t < s $, then go to the next iteration of step 4 
        \item \textcolor{blue}{cover the edge $\{x,y\}$ with the clique $C_l$ and update $\R$}  \textcolor{teal}{// Figure \ref{fig:ecc_algo2_subs} (a)}
        \item $\Q \leftarrow \textbf{ACCS2}(G, t-s, \C, \R)$
        \item if $\Q \ne \emptyset$, then return $\Q$
        \item \textcolor{blue}{undo changes done to $\C$ and $\R$ at step 4c} \textcolor{teal}{// Figure \ref{fig:ecc_algo2_subs} (b)}
    \end{enumerate}
    \item if $t \geq 2$, then
    \begin{enumerate}
        \item \textcolor{blue}{set $\C$ to $\C \cup \{\{x,y\}\}$ and update $\R$} \textcolor{teal}{// Figure \ref{fig:ecc_algo2_subs} (c)}
        \item $\Q \leftarrow \textbf{ACCS2}(G, t-2, \C, \R)$
        \item if $\Q \ne \emptyset$, then return $\Q$
        \item \textcolor{blue}{undo changes done to $\C$ and $\R$ at step 5a} \textcolor{teal}{// Figure \ref{fig:ecc_algo2_subs} (d)}
    \end{enumerate}
    \item return $\emptyset$
    \end{enumerate} 
}}   
    \caption{A bounded search tree algorithm for \emph{ACC}, denoted \emph{ACCS2}.}
    \label{fig:acc_algo2}
\end{figure}

In a search tree of \emph{ACCS2}, there is no depth bound for the number of cliques. \emph{ACCS} will create as many cliques as permitted by the parameter $t$, until it finds a clique cover of $G$ at step 1. At step 4a of \emph{ACCS2}, we count out of two vertices $x$ and $y$ how many are missing in $C_l$: it could only be one or two. Only scenario for which we cannot execute steps 4c-4f would be when $s = 2$ but $t = 1$. We are reusing the subroutines of \emph{ECCS2} for \emph{ACCS2}, and rest of the steps of \emph{ACCS2} are fairly straightforward adaptation of \emph{ECCS2}. We can adapt the proof of Lemma \ref{lemma_eccs2_ok} to show correctness of \emph{ACCS2} (see Appendix \ref{sec:app_algos}). Thus we conclude the following.

\begin{restatable}{lemma}{LemmaACCSTwoOK}
\label{lemma_accs2_ok}
Algorithm \emph{ACCS2} correctly solves the parameterized problem \emph{ACC}.
\end{restatable}

\begin{lemma}
\label{lemma_accs2_nn}
The number of nodes in a search tree of \emph{ACCS2} is at most $4^{t \log t}$.
\end{lemma}
\begin{proof}
The branching factors of \emph{ACCS2} are bounded by $\binom{d+1}{2}$: this follows from the first part of the proof of Lemma \ref{lemma_eccs2_bf}. By Rule \ref{rule_acc_4}, $d+1 \leq n \leq t$. The depth of a search tree of \emph{ACCS2} is at most $t$. Therefore, the number of nodes in a search tree of \emph{ACCS2} is at most $\binom{t}{2}^t < t^{2t} = 4^{t\log t}$.
\end{proof}

Theorem \ref{thm_acc2} follows from Lemma \ref{lemma_accs2_nn} and Proposition \ref{prop_eqv1}, considering time needed for data reduction. The running time of Corollary \ref{cor_acc} follows from similar arguments presented for Corollary \ref{cor_acc2}. The space bound of Corollary \ref{cor_acc} follows from Lemma \ref{lemma_eccs2_space}, since for an \emph{assignment-minimum clique cover} $k = m$ suffices.

\subsection{Weighted Edge Clique Partition}

We describe a bounded search tree algorithm for the \emph{WECP} problem, considering that a given instance of \emph{WECP} would be reduced with respect to the data reduction rules described by \cite{feldmann2020fixed}. For \emph{WECP}, \cite{feldmann2020fixed} have developed a \emph{bi-kernel} of size at most $4^k$. The \emph{bi-kernel} is based on a general problem that considers a subset of vertices annotated with integer weights as input, in addition to the input of \emph{WECP}. The general problem is as follows.

\begin{center}
\fbox{\parbox{6.0in}{
\emph{Annotated weighted edge clique partition} (\emph{AWECP})

\textbf{Input:} A graph $G=(V,E)$, a weight function on edges $w^E: E \rightarrow \mathbb{Z}_{>0}$, a nonnegative integer $k$, a set of vertices $S \subseteq V$, and a weight function $w^S: S \rightarrow \mathbb{Z}_{>0}$.

\textbf{Output:} If one exists, a clique cover $\C$ of $G$ such that (1) $|\C| \leq k$, (2) each edge $e \in E$ appears in exactly $w^E(e)$ cliques of $\C$, and (3) each vertex $x \in S$ appears in exactly $w^S(x)$ cliques of $\C$; otherwise report \emph{NO}.
 }}    
\end{center}

When $S = \emptyset$, \emph{AWECP} reduces to the special case \emph{WECP}. For the data reduction rules that give rise to an instance of \emph{AWECP} from an instance of \emph{WECP}, see Section 2 of \cite{feldmann2020fixed}. We point out that a reduced instance of \emph{AWECP} with respect to the data reduction rules of \cite{feldmann2020fixed} does not contain any isolated vertices.

Note that for any vertex $x \in S$, if the weight on the vertex $x$ is smaller than the weight of all the edges incident on $x$, i.e., $w^S(x) < w^E(e)$ for all $e = \{x,y\} \in E$, then we immediately know that a given instance of \emph{AWECP} is a \emph{NO} instance. Also, for any vertex $x \in S$, if the weight on the vertex $x$ is larger than the sum of the weights of all the edges incident on $x$, i.e., $w^S(x) > \sum_{e= \{x,y\} \in E} w^E(e)$, then we immediately know that a given instance of \emph{AWECP} is a \emph{NO} instance.

\begin{figure}
    \centering
    \fbox{\parbox{6.0in}{
\underline{Annotated-Weighted-Edge-Clique-Partition-Search$(G, k, w^E, S, w^S, \C, \R)$:}

\textcolor{teal}{\\ // Abbreviated \textbf{AWECPS}$(G, k, w^E, S, w^S, \C,\R)$}

    \begin{enumerate}
        \item if $w^E(e) = 0$ for all $e \in E$ and $w^S(z) =0 $ for all $z \in S$, then return $\C$
        
        \textcolor{teal}{// cf. step 1 \emph{ECCS2} Figure \ref{fig:ecc_algo2}}
        
        \item select the \emph{last} edge $e=\{x,y\}$ from the \emph{DEP} of $G$ where $y \in N_d(x)$ and $w^E(e) > 0$
        
        \textcolor{teal}{// cf. step 2 \emph{ECCS2} Figure \ref{fig:ecc_algo2}}
        \item for each $l \in R_x \cap R_y$ such that $\{x,y\} \not\subseteq C_l$ do 
    \begin{enumerate}
        \item \textcolor{blue}{if '$e$ cannot be included' in $C_l$, then go to the next iteration of step 3} \textcolor{teal}{// Figure \ref{fig:awecp_algo_subs} (a)}
        \item \textcolor{blue}{include the edge $\{x,y\}$ in the clique $C_l$, and update $\R$, $w^E$, $w^S$}   \textcolor{teal}{// Figure \ref{fig:awecp_algo_subs} (b)}
        \item $\Q \leftarrow \textbf{AWECPS}(G, k, w^E, S, w^S, \C, \R)$
        \item if $\Q \ne \emptyset$, then return $\Q$
        \item \textcolor{blue}{undo changes done to $\C$, $\R$, $w^E$, and $w^S$ at step 3b} \textcolor{teal}{// Figure \ref{fig:awecp_algo_subs} (c)}
    \end{enumerate}
    \item if $k > 0$, then
    \begin{enumerate}
        \item if ($x \in S$ and $w^S(x) = 0$) or ($y\in S$ and $w^S(y) =0$), then return $\emptyset$
        \item \textcolor{blue}{set $\C$ to $\C \cup \{\{x,y\}\}$, and update $\R$, $w^E$, $w^S$} \textcolor{teal}{// Figure \ref{fig:awecp_algo_subs} (d)}
        \item $\Q \leftarrow \textbf{AWECPS}(G, k-1, w^E, S, w^S, \C, \R)$
        \item if $\Q \ne \emptyset$, then return $\Q$
        \item \textcolor{blue}{undo changes done to $\C$, $\R$, $w^E$, and $w^S$ at step 4b} \textcolor{teal}{// Figure \ref{fig:awecp_algo_subs} (e)}
    \end{enumerate}
    \item return $\emptyset$
    \end{enumerate} 
}}   
    \caption{A bounded search tree algorithm for \emph{AWECP}, denoted \emph{AWECPS}.}
    \label{fig:awecp_algo}
\end{figure}

\begin{figure}[!hbtp]
    \centering
    \begin{minipage}{0.45\textwidth}
    \begin{subfigure}{\linewidth}
        \fbox{\parbox{\linewidth}{
        \textcolor{blue}{\underline{Step 3a:}}
        
        \begin{enumerate}
            \item if $x \in S$ and $x \not\in C_l$ and $w^S(x) = 0$, then report \emph{'$e$ cannot be included'}
            \item if $y \in S$ and $y \not\in C_l$ and $w^S(y) = 0$, then report \emph{'$e$ cannot be included'}
            \item for each $z \in C_l \backslash \{x,y\}$ do
            \begin{itemize}
                \item Let $e_1 = \{x,z\}$ and $e_2 = \{y,z\}$
            \end{itemize}
            \begin{enumerate}
                \item if ($x \not \in C_l$ and $w^E(e_1) = 0$) or ($y \not\in C_l$ and $w^E(e_2) = 0$), then report \emph{'$e$ cannot be included'}
            \end{enumerate}
            
            \item report \emph{'$e$ can be included'}
        \end{enumerate}
        }}  
    \caption{}
    \end{subfigure}
    
    \begin{subfigure}{\linewidth}
        \fbox{\parbox{\linewidth}{
        \textcolor{blue}{\underline{Step 3b:}}
        
        \begin{enumerate}
            \item \textcolor{blue}{execute the subroutine for step 3a of \emph{ECCS2} from Figure \ref{fig:ecc_algo2_subs}}
            \item for each $z \in C_l \backslash \{x,y\}$ do 
            \begin{enumerate}
                \item if $x_l = 0$, then decrement $w^E(e_1)$ by $1$, where $e_1 = \{x,z\}$
                \item if $y_l = 0$, then decrement $w^E(e_2)$ by $1$, where $e_2 = \{y,z\}$
            \end{enumerate}
            \item decrement $w^E(e)$ by $1$
            \item if $x_l = 0$ and $x \in S$, then decrement $w^S(x)$ by $1$
            \item if $y_l = 0$ and $y \in S$, then decrement $w^S(y)$ by $1$
        \end{enumerate}
        }}  
    \caption{}
    \end{subfigure}

    \end{minipage}
    \hspace{\fill}
    \begin{minipage}{0.45\textwidth}
    
        \begin{subfigure}{\linewidth}
            \fbox{\parbox{\linewidth}{
            \textcolor{blue}{\underline{Step 3e:}}
            
            \begin{enumerate}
                \item \textcolor{blue}{execute the subroutine for step 3b of \emph{ECCS2} from Figure \ref{fig:ecc_algo2_subs}}
                \item for each $z \in C_l \backslash \{x,y\}$ do
                \begin{enumerate}
                    \item if $x_l = 0$, then increment $w^E(e_1)$ by $1$, where $e_1 = \{x,z\}$ 
                    \item if $y_l = 0$, then increment $w^E(e_2)$ by $1$, where $e_2 = \{y,z\}$
                \end{enumerate}
                \item increment $w^E(e)$ by $1$
                \item if $x_l = 0$ and $x \in S$, then increment $w^S(x)$ by $1$
                \item if $y_l = 0$ and $y \in S$, then increment $w^S(y)$ by $1$
            \end{enumerate}
            }}  
        \caption{}
        \end{subfigure}
        
         \begin{subfigure}{\linewidth}
        \fbox{\parbox{\linewidth}{
        \textcolor{blue}{\underline{Step 4b:}}
        
        \begin{enumerate}
            \item \textcolor{blue}{execute the subroutine for step 4a of \emph{ECCS2} from Figure \ref{fig:ecc_algo2_subs}}
            \item decrement $w^E(e)$ by $1$
            \item if $x \in S$, then decrement $w^S(x)$ by $1$
            \item if $y \in S$, then decrement $w^S(y)$ by $1$
        \end{enumerate}
        }}  
    \caption{}
    \end{subfigure}
    
     \begin{subfigure}{\linewidth}
        \fbox{\parbox{\linewidth}{
        \textcolor{blue}{\underline{Step 4e:}}
        
        \begin{enumerate}
            \item \textcolor{blue}{execute the subroutine for step 4d of \emph{ECCS2} from Figure \ref{fig:ecc_algo2_subs}}
            \item increment $w^E(e)$ by $1$
            \item if $x \in S$, then increment $w^S(x)$ by $1$
            \item if $y \in S$, then increment $w^S(y)$ by $1$
        \end{enumerate}
        }}  
    \caption{}
    \end{subfigure}
    \end{minipage}
    \caption{Subroutines of \emph{AWECPS} shown in Figure \ref{fig:awecp_algo}.}
    \label{fig:awecp_algo_subs}
\end{figure}

Figure \ref{fig:awecp_algo} shows a bounded search tree algorithm for \emph{AWECP}, henceforth referred as \emph{AWECPS}. The subroutines for steps 3a, 3b, 3e, 4a, 4d of \emph{AWECPS} are described in Figure \ref{fig:awecp_algo_subs}. \emph{AWECPS} extends the mechanism of building clique cover employed by\emph{ECCS2}. At every node of a search tree, \emph{AWECPS} extends a clique cover if and only if the extension is (locally) consistent with the weight functions $w^E$ and $w^S$.

\emph{AWECPS} treats the weight functions $w^E$ and $w^S$ as budgets for the edges of $E$ and vertices of $S$ respectively. At every node of a search tree, every edge $e$ has a budget $w^E(e)$, denoting how many more cliques the edge $e$ can be included. Similarly, at every node of a search tree, every vertex $z \in S$ also has a budget $w^S(z)$, denoting how many more cliques the vertex $z$ can be included. For \emph{AWECPS}, we maintain the notions of including an edge (or a vertex) in a clique (in this context covering an edge $e$ by cliques could be used to denote including the edge $e$ in exactly $w^E(e)$ distinct cliques).

Since both $w^E$ and $w^S$ are positive integer weight functions, if budgets of all edges of $E$ and budgets of all vertices of $S$ can be spent exactly at a node of a search tree of \emph{AWECPS} with at most $k$ cliques, then the given instance of \emph{AWECP} is a \emph{YES} instance, and step 1 would return a corresponding clique cover. Step 2 of \emph{AWECPS} essentially select the same edge $e=\{x,y\}$ as would be selected by step 2 of \emph{ECCS2}, except \emph{AWECPS} selects the same edge $e$ repeatedly until its budget is entirely exhausted, i.e., $w^E(e) = 0$. Steps 3a and 4a enforce that including an edge $e=\{x,y\}$ to a clique must be within the budgets of the corresponding edges of $E$ and within the budgets of corresponding vertices of $S$. Rest of the steps of \emph{AWECPS} are self-explanatory.

Turning to the subroutines of \emph{AWECPS} in Figure \ref{fig:awecp_algo_subs}, we see that the subroutine for step 3a simply makes sure that every new vertex or new edge that would be included in the clique $C_l$ (as a result of including the edge $e=\{x,y\}$ in the clique $C_l$) has the budget to do so. Rest of the subroutines make sure that the budgets of the edges of $E$ and the vertices of $S$ are updated accordingly when we include (resp. exclude) an edge $e$ in (resp. from) a clique. Note that (for brevity) the subroutines of \emph{ECCS2} are invoked from the subroutines of \emph{AWECPS}. 

The proof correctness of \emph{ECCS2} (Lemma \ref{lemma_eccs2_ok}) can be adapted to obtain a proof of correctness for \emph{AWECPS} (see Appendix \ref{sec:app_algos}). Thus we conclude the following.

\begin{restatable}{lemma}{LemmaAWECPSok}
\label{lemma_awecp_ok}
Algorithm \emph{AWECPS} correctly solves the parameterized problem \emph{AWECP}.
\end{restatable}

\begin{lemma}
\label{lemma_awecp_nn}
The number of nodes in a search tree of \emph{AWECPS} is at most $2^{\beta k \log k}$.
\end{lemma}
\begin{proof}
It is straightforward to see that the depth of a search tree of \emph{AWECPS} is bounded by $\beta k$. Let $w$ be the maximum edge weight in an instance of \emph{AWECP}, i.e., $w = \max_{e \in E} \{w^E(e)\}$. Then, number of branches at any node of a search tree of \emph{AWECPS} is bounded by $\min\{k, \binom{d+1}{2}w\} \leq k$: this follows from similar arguments presented for \emph{ECCS2} in Lemma \ref{lemma_eccs2_bf}. Therefore, the number of nodes in a search tree of \emph{AWECPS} is at most $k^{\beta k} = 2^{\beta k \log k}$.
\end{proof}

Considering time needed for data reduction, Theorem \ref{thm_wecp} follows from Lemma \ref{lemma_awecp_nn} and Proposition \ref{prop_eqv1}.
 
For an instance of \emph{WECP}, the algorithm of \cite{feldmann2020fixed} takes $2^{O(k^{3/2}w^{1/2}\log (k/w))}n^{O(1)}$ time, where $w$ is the maximum edge weight of the instance. For $\beta = o((\sqrt{kw}\log (k/w)) / \log k)$, \emph{AWECPS} improves the running time by a factor of $\frac{2^{O(k^{3/2}w^{1/2}\log (k/w))}}{2^{\beta k \log k}} = 2^{O(k(\sqrt{kw}\log (k/w) - \beta \log k))} = 2^{O(k^{3/2}w^{1/2}\log (k/w))}$. This is a significant improvement, considering the fact that $\beta$ is bounded by the graph size, while $k$ and $w$ could be arbitrarily large. For nontrivial instances of \emph{WECP}, $k$ could be up to (but not including) $mw$.

On the other hand, if the maximum edge weight $w$ is bounded by some constant, then for $\beta = o(\sqrt{k})$, \emph{AWECPS} is $\frac{2^{O(k^{3/2} \log k)}}{2^{\beta k \log k}} = 2^{O(k \log k(\sqrt{k} - \beta ))} = 2^{O(k^{3/2} \log k)}$ times faster than the algorithm of \cite{feldmann2020fixed}.

\subsection{Exact Weighted Clique Decomposition}

The data reduction rules of \emph{WECP} described by \cite{feldmann2020fixed} are also applicable to any instance of \emph{EWCD}. Consequently, we have a \emph{bi-kernel} for \emph{EWCD} with $4^k$ vertices. The \emph{bi-kernel} is based on a general problem that considers a subset of vertices annotated with positive weights as input, in addition to the input of \emph{EWCD}. The general problem is as follows.

\begin{center}
\fbox{\parbox{6.0in}{
\emph{Annotated exact weighted clique decomposition} (\emph{AEWCD})

\textbf{Input:} A graph $G=(V,E)$, a weight function on edges $w^E: E \rightarrow \mathbb{R}_{>0}$, a nonnegative integer $k$, a set of vertices $S \subseteq V$, and a weight function $w^S: S \rightarrow \mathbb{R}_{>0}$.

\textbf{Output:} If one exists, a clique cover $\C$ of $G$ and positive weight $\gamma_i$ for every clique $C_i \in \C$ such that (1) $|\C| \leq k$, (2) for each edge $e= \{x,y\} \in E$, $\sum_{\{x,y\} \in C_i, C_i \in \C} \gamma_i = w^E(e)$, and (3) for each vertex $x \in S$, $\sum_{x \in C_i, C_i \in \C} \gamma_i = w^S(x)$; otherwise report \emph{NO}.
 }}    
\end{center}

Note that when $S = \emptyset$, \emph{AEWCD} reduces to the special case \emph{EWCD}. If the weight functions are integer valued and weight on each of the cliques are forced to be $1$, then \emph{AEWCD} is equivalent to \emph{AWECP}.

Let $\C$ be a clique cover of $G$. If $\C$ contains at most $k$ cliques, then, for an instance $(G,k,w^E,S,w^S)$ of \emph{AEWCD}, the following linear program (LP) would compute weights of the cliques $\gamma = \{\gamma_1, \gamma_2, \ldots, \gamma_{|\C|} \}$ in $\C$ (if the LP is feasible). Note that the objective function of the LP is a constant, i.e., an LP solver would only need to find a feasible solution $\gamma = \{\gamma_1, \gamma_2, \ldots, \gamma_{|\C|}\}$, if one exists.

\begin{center}
\fbox{\parbox{4.5in}{
    \textcolor{black}{
    \begin{equation}\tag{\textcolor{blue}{A}}
        \begin{array}{lll}
        \text{minimize}  & \displaystyle\sum\limits_{i: C_i \in \C} 0 \times \gamma_i  \\
        \text{subject to}& \displaystyle\sum\limits_{i: \{x,y\} \in C_i, C_i \in \C} \gamma_i  = w^E(e),& \text{for all } e = \{x,y\} \in E\\
        & \displaystyle\sum\limits_{i: x \in C_i, C_i \in \C} \gamma_i  = w^S(x),& \text{for all } x \in S\\
          & \gamma_i \geq 0, &\text{for all } i:  C_i \in \C
        \end{array}
    \label{aewcd_lp}
    \end{equation}
}}}    
\end{center}

\begin{figure}[!hbtp]
    \centering
    \fbox{\parbox{5.0in}{
\underline{Annotated-Exact-Weighted-Clique-Decomposition-Search$(G, k, \C, \R)$:}

\textcolor{teal}{\\ // Abbreviated \textbf{AEWCDS}$(G,k,\C,\R)$}

    \begin{enumerate}
        \item if $\C$ covers edges of $G$, then 
        \begin{enumerate}
            \item for $\C$ solve the LP (\textcolor{blue}{A})
            \item if (\textcolor{blue}{A}) is feasible, then return $(\C, \gamma)$
            \item else return $(\emptyset, \emptyset)$
        \end{enumerate}
        \item select the \emph{last} uncovered edge $\{x,y\}$ from the \emph{DEP} of $G$ where $y \in N_d(x)$
        
        \item for each $l \in R_x \cap R_y$ such that $\{x,y\} \not\subseteq C_l$ do 
    \begin{enumerate}
        \item \textcolor{blue}{include the edge $\{x,y\}$ in the clique $C_l$ and update $\R$}  \textcolor{teal}{// Figure \ref{fig:ecc_algo2_subs} (a)}
        \item mark the edge $\{x,y\}$ as covered 
        \item $(\Q, \gamma) \leftarrow \textbf{AEWCDS}(G, k, \C, \R)$
        \item if $\Q \ne \emptyset$, then return $(\Q, \gamma)$
        \item mark the edge $\{x,y\}$ as uncovered 
        \item $(\Q, \gamma) \leftarrow \textbf{AEWCDS}(G, k, \C, \R)$
        \item if $\Q \ne \emptyset$, then return $(\Q, \gamma)$
        \item \textcolor{blue}{undo changes done to $\C$ and $\R$ at step 3a} \textcolor{teal}{// Figure \ref{fig:ecc_algo2_subs} (b)}
    \end{enumerate}
    \item if $k > 0$, then
    \begin{enumerate}
        \item \textcolor{blue}{set $\C$ to $\C \cup \{\{x,y\}\}$ and update $\R$} \textcolor{teal}{// Figure \ref{fig:ecc_algo2_subs} (c)}
        \item mark the edge $\{x,y\}$ as covered
        \item $(\Q, \gamma) \leftarrow \textbf{AEWCDS}(G, k-1, \C, \R)$
        \item if $\Q \ne \emptyset$, then return $(\Q, \gamma)$
        \item mark the edge $\{x,y\}$ as uncovered 
        \item $(\Q, \gamma) \leftarrow \textbf{AEWCDS}(G, k-1, \C, \R)$
        \item if $\Q \ne \emptyset$, then return $(\Q, \gamma)$
        \item \textcolor{blue}{undo changes done to $\C$ and $\R$ at step 4a} \textcolor{teal}{// Figure \ref{fig:ecc_algo2_subs} (d)}
    \end{enumerate}
    \item return $(\emptyset, \emptyset)$
    \end{enumerate} 
}}   
    \caption{A bounded search tree algorithm for \emph{AEWCD}, denoted \emph{AEWCDS}.}
    \label{fig:aewcd_algo}
\end{figure}

Figure \ref{fig:aewcd_algo} shows a bounded search tree algorithm for \emph{AEWCD}, henceforth referred to as \emph{AEWCDS}. We reuse the subroutines of \emph{ECCS2} for \emph{AEWCDS}: steps 3a, 3h, 4a, 4h of \emph{AEWCDS} are identical to the steps 3a, 3d, 4a, 4d of \emph{ECCS2} described in Figure \ref{fig:ecc_algo2_subs}. We assume \emph{AEWCDS} has access to the weight functions at step 1 of any node of a search tree. Whenever a clique cover $\C$ of $G$ is found at a node, we check the feasibility of $\C$ for the instance $(G,k,w^E,S,w^S)$, by solving the LP (\textcolor{blue}{A}). If (\textcolor{blue}{A}) is feasible, then the clique cover $\C$ and a corresponding solution $\gamma = \{\gamma_1, \gamma_2, \ldots\}$ of (\textcolor{blue}{A}) is returned at step 1.

For \emph{AEWCD}, we explicitly distinguish between covering an edge by a clique and including an edge in a clique. An edge can be included in many cliques, without being identified as covered. If an edge $\{x,y\}$ is marked as covered (steps 3b, 4b) at a node $u^T$ of a search tree $T$, then the corresponding descendent nodes (steps 3c, 4c) of $u^T$ in $T$ would not select the edge $\{x,y\}$ as uncovered at step 2. On the other hand, if an edge $\{x,y\}$ is marked as uncovered (steps 3e, 4e) at a node $u^T$ of a search tree $T$, then the corresponding descendent nodes (steps 3f, 4f) of $u^T$ in $T$ may select the edge $\{x,y\}$ as uncovered at step 2.

\begin{figure}[!hbtp]
    \centering
    \resizebox{6.5cm}{5.0cm}{%
         \tikzset{main node/.style={text=blue, circle,fill=black!20,draw,minimum size=0.8cm,inner sep=0pt},
            }
 \begin{tikzpicture}
    \node[main node] (1) {$a$};
    \node[main node] (2) [right = 3.5cm of 1] {$b$};
    \node[main node] (3) [below = 3.5cm of 2] {$c$};
    \node[main node] (4) [below = 3.5cm of 1] {$d$};

    \path[draw,thick]
    (1) edge node[above] {$101$} (2)
    (1) edge node[above left =0.8cm and 0.3cm] {$2$} (3)
    (1) edge node[left] {$1$} (4)
    (2) edge node[right] {$2$} (3)
    (2) edge node[below left=0.4cm and 0.8cm] {$1$} (4)
    (3) edge node[below] {$1$} (4)
    
    ;
\end{tikzpicture} 
    }
    \caption{An example graph $G$ with positive weights on edges. For $k=3$, $G$ with the weights on the edges (and $S=\emptyset$) is a \emph{YES} instance of \emph{AEWCD}.
    Cliques of a corresponding clique cover $\C$ of $G$, and feasible weights on the cliques are as follows:
    \textcolor{blue}{$C_1 = \{a,b,c,d\}$, $C_2 = \{a,b,c\}$, $C_3 = \{a,b\}$},
    \textcolor{blue}{$\gamma_1 = 1$, $\gamma_2 = 1$, $\gamma_3 = 99$.}
     }
    \label{fig:ewcd_example}
\end{figure}
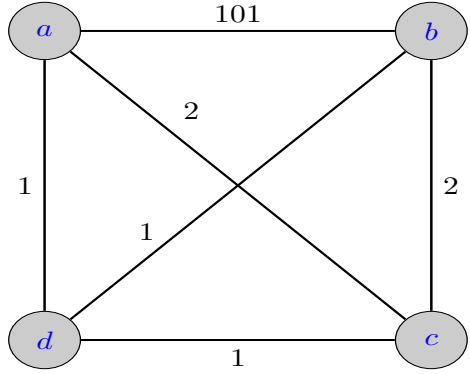

For \emph{AEWCDS}, the branches at steps 3e-3g and 4e-4g are crucial: an edge $e=\{x,y\}$ may need to be included in many cliques so that the weight $w(e)$ can be matched with the weights of the neighbouring edges of $e$. To see this, consider the complete graph $K_4$ and weights on the edges as shown in Figure \ref{fig:ewcd_example}. For $k=3$, the instance of \emph{AEWCD} is a \emph{YES} instance. A corresponding clique cover $\C$ of $G$ consists of following cliques: $C_1 = \{a,b,c,d\}$, $C_2 = \{a,b,c\}$, $C_3 = \{a,b\}$. Corresponding feasible weights on the cliques are $\gamma_1 = 1$, $\gamma_2 = 1$, and $\gamma_3 = 99$. WLOG, we can assume that the edge $e=\{a,b\}$ is the last edge in the \emph{DEP} of $G$. The edge $e$ needs to appear in all three cliques of $\C$. Therefore, we need to make sure that an edge is included in a sufficient number of cliques (at most $k$), so that we have correct combination of cliques for the edges that are chosen later from the edge permutation.

The way we have described \emph{AEWCDS} (in Figure \ref{fig:aewcd_algo}) is merely for ease of reading. An equivalent (in terms of correctness) description would have a subroutine consisting of step 3 and step 4, excluding steps 3b, 3e-3g, 4b, 4e-4g. Now, instead of step 3 and 4, we can simply call the subroutine once, after marking the edge $\{x,y\}$ as covered, and if the call fails, then we can call the subroutine once, after marking the edge $\{x,y\}$ as uncovered. The search trees corresponding to the equivalent description are likely to be smaller for many instances of \emph{AEWCD}.

\cite{cooley2021parameterized} have described an LP based algorithm for \emph{AEWCD}, called \emph{CliqueDecomp-LP}. The LP used by \emph{CliqueDecomp-LP} has $k$ variables and at most $4k^2$ constraints. In contrast, the LP (\textcolor{blue}{A}) has at most $k$ variables and at most $m+|S| \leq m + n$ constraints. At every iteration, \emph{CliqueDecomp-LP} selects a permutation matrix $P \in \{0,1\}^{2k \times k}$. For every such permutation matrix $P$, \emph{CliqueDecomp-LP} needs to solve the corresponding LP, regardless of whether $P$ correspond to a clique cover of the graph or not. For a single permutation matrix $P$, \emph{CliqueDecomp-LP} may need to solve the LP up to $2k$ times. Therefore, \emph{CliqueDecomp-LP} may need to solve the LP up to $2^{2k^2} \times 2k$ times, in the worst case. In contrast, \emph{AEWCDS} solves the LP (\textcolor{blue}{A}) (at step 1) only when $\C$ is a clique cover of the graph. In most cases, for \emph{YES} instances of \emph{AEWCD}, the number of clique covers of a graph with at most $k$ cliques is significantly smaller than $2^{2k^2} \times 2k$.

A correctness proof for \emph{AEWCDS} can be obtained similarly as the correctness proof for \emph{AWECPS} (Lemma \ref{lemma_awecp_ok}). A minor modification would be required for the base case of induction. In the base case, i.e., $k=1$, for some constant $q$, we would have $w^E(e) = q$ for all $e \in E$ and $w^S(x) = q$ for all $x \in S$. Rest of the proof would be analogous. Thus we have the following.

\begin{lemma}
Algorithm \emph{AEWCDS} correctly solves the parameterized problem \emph{AEWCD}. 
\end{lemma}

\begin{lemma}
\label{lemma_aewcd_nn}
The number of nodes in a search tree of \emph{AEWCDS} is at most $2^{\beta k (1+ \log k)}$.
\end{lemma}
\begin{proof}
The depth of a seach tree of \emph{AEWCDS} is bounded by $\beta k$. At step 3 of any node of the search tree $R_x \cap R_y \leq k$. If $R_x \cap R_y = k$, then there is no branch at step 4 of the node. Therefore, number of branches at any node of a search tree of \emph{AEWCDS} is at most $2k$. It follows that the number of nodes in a search tree of \emph{AEWCDS} is at most $(2k)^{\beta k} = 2^{\beta k (1 + \log k)}$.
\end{proof}

For \emph{AEWCDS}, the LP (\textcolor{blue}{A}) would have at most $k$ variables, and at most $m+n$ constraints. For nontrivial instances of \emph{AEWCD}, we can assume $k < m$. Therefore, (\textcolor{blue}{A}) is solvable in $O(k^3L)$ time, where $L$ is the number of bits needed to encode the LP \cite{cohen2021solving, van2020deterministic}. Since $n \leq 4^k$, running time of \emph{AEWCDS} is $2^{O(\beta k \log k)}L$. Now, considering time needed for data reduction, Theorem \ref{thm_ewcd} follows from Proposition \ref{prop_eqv1}.

For the cases when all the weights are restricted to integers (including the weights on the cliques), \cite{cooley2021parameterized} have described an integer partitioning based algorithm called \emph{CliqueDecomp-IP}. \emph{CliqueDecomp-IP} requires maintaining $(w+1)^k$ weight matrices in the worst case, where $w$ is the maximum edge weight. Consequently, space requirement of \emph{CliqueDecomp-IP} could be prohibitive for many \emph{AEWCD} instances even if the instances are solvable in a reasonable amount of time. For these cases, we point out that instead of solving the LP (\ref{aewcd_lp}), at step 1 of \emph{AEWCDS}, one can search for feasible set of weights $\gamma$, by enumerating at most $w^{k} = 2^{k \log w}$ choices. Therefore, \emph{AEWCD} instances restricted to integer weights are solvable in $2^{O(k (\beta \log k + \log w))}n^{O(1)}$ time, using $O(m+k \Delta)$ space (Lemma \ref{lemma_eccs2_space}).

\emph{CliqueDecomp-LP} and \emph{CliqueDecomp-IP} both require enumerating $2^{2k^2}$ permutation matrices in the worst case. For $\beta = o(k / \log k)$, $\emph{AEWCDS}$ improves the running time by a factor of $\frac{2^{O(k^2)}}{2^{O(\beta k \log k)}} = 2^{O(k(k - \beta \log k))} = 2^{O(k^2)}$.

\subsection{Generalized \emph{Vertex Clique Cover} and \emph{Colorability}}

We conclude our demonstration of applicability of our new framework with description of a bounded search tree algorithm for \emph{LRCC}. This serves several purposes. First, this highlights the natural bridges that exist between clique cover problems and corresponding graph coloring problems in the complement graph. Second, even though many variants of graph coloring problems had been studied in parameterized complexity (see \cite{fiala2011parameterized, jansen2019computing}), the results are not useful for obtaining FPT algorithms for the corresponding vertex clique cover problems (since no \emph{parameterized reductions} are known from parameterization of the vertex clique cover problems to the parameterizations of graph coloring problems that are known to be FPT). Third, we are able to show that instances of \emph{PMC} can be solved by a \emph{parameterized reduction} from \emph{PMC} to \emph{LRCC}.

The equivalence of \emph{VCC} and \emph{Colorability} stems from the fact that a clique in a graph corresponds to an independent set in the complement graph, therefore, a partition of vertices of a graph into cliques corresponds to a partition of vertices into independent sets in the complement graph. Since \emph{Colorability} is \emph{NP-complete} for $k \geq 3$ \cite{stockmeyer1973planar}, the existence of an $f(k)n^{O(1)}$ time algorithm for \emph{VCC} would imply $P=NP$. Therefore, \emph{LRCC} is not FPT with respect to $k$ unless $P=NP$. We show that \emph{LRCC} is FPT with respect to $\beta$ and $k$.

\begin{figure}[!hbtp]
    \centering
    \resizebox{8.5cm}{4.0cm}{%
         \tikzset{main node/.style={text=blue, circle,fill=black!20,draw,minimum size=0.8cm,inner sep=0pt},
            }
 \begin{tikzpicture}
    \node[main node] (1) {$a$};
    \node[main node] (2) [above right = 2.0cm and 2.0cm of 1] {$b$};
    \node[main node] (3) [below right = 2.0cm and 2.0cm of 1] {$c$};
    \node[main node] (4) [right = 4.0cm of 1] {$d$};
    \node[main node] (5) [above right = 2.0cm and 2.0cm of 4] {$e$};
    \node[main node] (6) [below right = 2.0cm and 2.0cm of 4] {$f$};
    \node[main node] (7) [right = 4.0cm of 4] {$g$};

    \path[draw,thick]
    (1) edge node {} (2)
    (1) edge node {} (3)
    (2) edge node {} (3)
    (2) edge node {} (4)
    (3) edge node {} (4)
    (4) edge node {} (5)
    (4) edge node {} (6)
    (5) edge node {} (7)
    (5) edge node {} (6)
    (6) edge node {} (7)    
    ;
\end{tikzpicture} 
    }
    \caption{An example graph $G$. Let $E^*= \{\{b,d\}\}$ and $k=3$. The instances $(G,k,E^*)$ and $(G,k)$ are \emph{YES} instances of \emph{LRCC} and \emph{VCC} respectively. A corresponding solution $\C$ of $(G,k)$ has following cliques:
    \textcolor{blue}{$C_1 = \{a,b,c\}$, $C_2 = \{d,e,f\}$, $C_3 = \{g\}$}, missing the edge $\{b,d\}$ in all the cliques. A corresponding solution $\C^*$ of $(G,k,E^*)$ has following cliques: \textcolor{blue}{$C_1^* = \{a\}$, $C_2^* = \{b,c,d\}$, $C_3^* = \{e,f,g\}$.}
     }
    \label{fig:LRCC_example}
\end{figure}
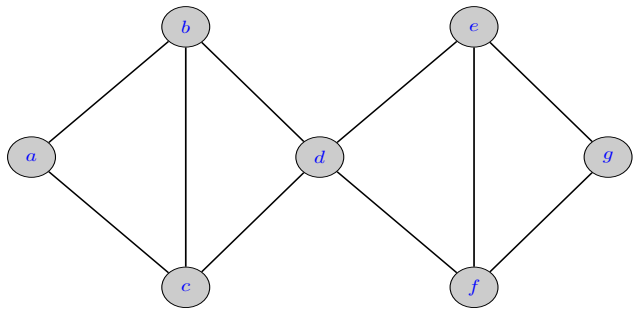

To solve an instance $I^*=(G,k,E^*)$ of \emph{LRCC} such that the number of edges in $E^*$ is small (far from being an instance of \emph{ECC}), one may think of solving an instance $I=(G,k)$ of \emph{VCC}, and then somehow convert the solution of $I$ to a solution of $I^*$. It is unlikely that there is any feasible approach to do that. To see this, consider the example shown in Figure \ref{fig:LRCC_example}. A solution of $(G,k)$ may miss many edges in $E^*$, even when the solutions of $I$ and $I^*$ need same number of cliques. Therefore, we turn to other viable approaches.

\begin{figure}[!hbtp]
    \centering
    \fbox{\parbox{5.8in}{
\underline{Link-Respected-Clique-Cover-Search$(G, k, E^*, \C, \R)$:}

\textcolor{teal}{\\ // Abbreviated \textbf{LRCCS}$(G,k, E^*, \C,\R)$}

    \begin{enumerate}
        \item if $\C$ covers vertices of $G$ and edges of $E^*$, then return $\C$
        
        \item select the \emph{last} uncovered edge $e = \{x,y\} \in E^*$ from the \emph{DEP} of $G$ where $y \in N_d(x)$

        \item select the \emph{last} uncovered vertex $u$ from the degeneracy ordering of $V$.
        
        \item if $e \ne \emptyset$ and $u \ne \emptyset$ and $x$ precedes $u$ in the degeneracy ordering of $V$, then set $e$ to be $\emptyset$

        \item if $e \ne \emptyset$ then set $R_x^*$ to be $R_x \cap R_y$
        \item if $e = \emptyset$ then set $x$ to be $u$ and $R_x^*$ to be $R_u$
        
        \item for each $l \in R_x^*$ do 
    \begin{enumerate}
        \item if $e = \emptyset$ then let $y$ be an arbitrary vertex of $C_l$
        \item \textcolor{blue}{cover the edge $\{x,y\}$ with the clique $C_l$ and update $\R$}  \textcolor{teal}{// Figure \ref{fig:ecc_algo2_subs} (a)}
        \item $\Q \leftarrow \textbf{LRCCS}(G, k, E^*, \C, \R)$
        \item if $\Q \ne \emptyset$, then return $\Q$
        \item \textcolor{blue}{undo changes done to $\C$ and $\R$ at step 3b} \textcolor{teal}{// Figure \ref{fig:ecc_algo2_subs} (b)}
    \end{enumerate}
    \item if $k > 0$, then
    \begin{enumerate}
        \item if $e = \emptyset$ then let $y$ to be $x$
        \item \textcolor{blue}{set $\C$ to $\C \cup \{\{x,y\}\}$ and update $\R$} \textcolor{teal}{// Figure \ref{fig:ecc_algo2_subs} (c)}
        \item $\Q \leftarrow \textbf{LRCCS}(G, k-1, E^*, \C, \R)$
        \item if $\Q \ne \emptyset$, then return $\Q$
        \item \textcolor{blue}{undo changes done to $\C$ and $\R$ at step 4b} \textcolor{teal}{// Figure \ref{fig:ecc_algo2_subs} (d)}
    \end{enumerate}
    \item return $\emptyset$
    \end{enumerate} 
}}   
    \caption{A bounded search tree algorithm for \emph{LRCC}, denoted \emph{LRCCS}.}
    \label{fig:LRCC_algo}
\end{figure}

One approach to solve \emph{LRCC} is to construct an \emph{edge clique cover} of a subgraph required to cover the edges of $E^*$, and then extend the \emph{edge clique cover} to a solution of $I^*$. This would only allow us to bound the branching factors of a search tree with $k$. We use a different approach so that we can preserve the bound on the branching factors of a search tree that we have obtained for \emph{ECCS2} (Lemma \ref{lemma_eccs2_bf}). Our approach is to interleave the tasks of covering the edges of $E^*$ and covering the vertices of $G$. Next, we elaborate on this approach.

Figure \ref{fig:LRCC_algo} shows our bounded search tree algorithm for \emph{LRCC}, henceforth referred as \emph{LRCCS}. \emph{LRCCS} achieves our goal of bounding the branching factors by selecting the first vertex $x$ in the reverse degeneracy ordering of $V$ such that either the vertex $x$ is uncovered or an edge $\{x,y\} \in E^*$ is uncovered where $y \in N_d(x)$. At any node of a search tree, \emph{LRCC} have two possible choices to consider: an uncovered edge $e= \{x,y\} \in E^*$ (step 2) or an uncovered vertex $u$ (step 3). The two possible choices give rise to four different cases: (1) $e=\emptyset$, (2) $u = \emptyset$, (3) $e \ne \emptyset$, $u \ne \emptyset$, and $x$ precedes $u$ in the degeneracy ordering of $V$, (4) $e \ne \emptyset$, $u \ne \emptyset$, and $x$ does not precede $u$ in the degeneracy ordering of $V$. Note that step 4 makes the recognition of cases (1) and (3) identical for the later steps.

For cases (1) and (3), the uncovered vertex $u$ needs to be covered (to be consistent with our goal). In these cases \emph{LRCC} chooses the \emph{representatives} in $R_u$ (step 6), to enumerate at step 7. For cases (2) and (4), the uncovered edge $e$ needs to be covered (to be consistent with our goal). In these cases \emph{LRCC} chooses the \emph{representatives} in $R_x \cap R_y$ (step 5), to enumerate at step 7. For cases (1) and (3), step 6a (resp. step 7a) converts the task of covering the vertex $x$ to covering an edge $\{x,y\}$ (resp. covering $ \{x\}=\{x,x\}$): this allows \emph{LRCCS} to reuse the subroutines of \emph{ECCS2} from Figure \ref{fig:ecc_algo2_subs}.

A proof of correctness of \emph{LRCC} can be obtained analogously as the proof of correctness shown for \emph{ECCS2} (Lemma \ref{lemma_eccs2_ok}). Thus we have the following.

\begin{lemma}
\label{lemma_lrcc_ok}
Algorithm \emph{LRCCS} correctly solves the parameterized problem \emph{LRCC}.
\end{lemma}

\begin{lemma}
\label{lemma_lrcc_nn}
The number of nodes in a search tree of \emph{LRCC} is at most $2^{\beta k \log k}$.
\end{lemma}
\begin{proof}
For an edge $e \ne \emptyset$ selected at step 2, the branching factors of a search tree of \emph{LRCC} is bounded by $\min\{k, \binom{d+1}{2}\}$: this follows from the same argument as presented for \emph{ECCS2} in Lemma \ref{lemma_eccs2_bf}. For a vertex $u \ne \emptyset$ selected at step 3, $N_d(u) \leq d$. Therefore, $|R_u| \leq \min \{k, \binom{d}{2}\}$. Combining preceding two bounds, the branching factors of a search tree of \emph{LRCCS} are bounded by $\min\{k, \binom{d+1}{2}\} \leq k$. It is straightforward to see that the depth of a search tree of \emph{LRCCS} is bounded by $\beta k$. Therefore, the number of nodes in a search tree of \emph{LRCC} is at most $k^{\beta k} = 2^{\beta k \log k}$.
\end{proof}

Since time spent at any node of a search tree of \emph{LRCCS} is bounded by a polynomial in $n$, Theorem \ref{thm_lrcc} follows from Lemma \ref{lemma_lrcc_nn}. We obtain a proof of Corollary \ref{cor_pmc} using a \emph{parameterized reduction} from \emph{PMC} to \emph{LRCC}.

\begin{proof}[Proof of Corollary \ref{cor_pmc}]
For an instance $(G,k,\F)$ of \emph{PMC}, in polynomial time we can construct the complement graph $\Bar{G}$. Now, the non-edges of $G$ in $\F$ becomes edges of $\Bar{G}$. Let $\alpha$ be the \emph{independence number} of $G$, and $\beta$ be the \emph{clique number} of $\Bar{G}$. Note that $\alpha = \beta$. By the equivalence of \emph{PMC} and \emph{LRCC} in the complement graph, $(G,k,\F)$ is a \emph{YES} instance of \emph{PMC} if and only if $(\Bar{G}, k, \F)$ is a \emph{YES} instance of \emph{LRCC}. The instance $(\Bar{G},k, \F)$ can be solved with the algorithm \emph{LRCCS} in $2^{\beta k \log k}n^{O(1)}$ time. Thus we can obtain a solution of $(G,k,\F)$ in $n^{O(1)} + 2^{\beta k \log k}n^{O(1)} = 2^{\alpha k \log k}n^{O(1)}$ time.
\end{proof}

\section{Implementations and Open Problems}
\label{sec_open}

To solve large real-world instances of a clique cover problem exactly, one may need to employ two things: effective data reduction rules and efficient search tree algorithm. In this article, utilizing a few data reduction rules, our focus has been on efficient search tree algorithms. Although, lower bounds on \emph{kernelization} (such as non-existence of polynomial kernel for \emph{ECC} with respect to clique cover size \cite{cygan2014clique}) restrict design of provably smaller kernels, such restrictions are barely impediment for designing data reduction rules that are effective in practice. An example of this is the work of \cite{strash2022effective} on \emph{vertex clique cover}. \emph{Vertex clique cover} does not admit any \emph{kernel} with respect to clique cover size (assuming $P \ne NP)$, but \cite{strash2022effective} have designed effective data reduction rules for \emph{vertex clique cover} that can solve large real-world instances. Therefore, designing effective data reduction rules for the clique cover problems we have studied is a direction that needs attention from future research.

To see efficacy of our proposed search tree algorithms in practice, we have implemented \emph{ECCG}, \emph{ECCS}, and \emph{ECCS2}, and have compared performance of the algorithms. Given the lack of effective data reduction rules for large real-world graphs, we have limited comparisons of the algorithms on random instances that \emph{ECCG} can solve within a few hours. The results are remarkable and demonstrate that improvements in our analyses do carry over into the implementations (code and details of the comparisons are available here: \url{https://drive.google.com/drive/folders/1ISa9PX0n5TeXFSwoNQnXzyoq63CZLMoZ}). On our test instances, \emph{ECCG} took approximately 2000-5200 seconds, whereas \emph{ECCS} and \emph{ECCS2} took approximately 1-26 seconds only. Within a span of few hours, \emph{ECCS} has been able to produce solutions using search trees that are orders of magnitude smaller than the corresponding search trees of \emph{ECCG}. On our test instances, the search trees of \emph{ECCG} have approximately 781-6200 million nodes, whereas the search trees of \emph{ECCS} have approximately 14-111 million nodes only.

We conclude by highlighting a few problems that we think are important to resolve.

Using Corollary \ref{cor_eccs2}, Corollary \ref{cor_eccs3}, and Corollary \ref{cor_eccs4}, we have shown that for sparse graphs running time of \emph{ECC} is better captured by degeneracy, instead of clique cover size. Thus we ask the following.

\begin{problem}
Does \emph{ECC} parameterized by $d$ has an FPT algorithm running in $f(d)n^{O(1)}$ time?
\end{problem}

For planar graphs, we have shown that \emph{ECC} is solvable in $2^{O(k)}n^{O(1)}$ time (Corollary \ref{cor_eccs5}). For many \emph{NP-complete} problems on planar graphs, the bidimensionality theory \cite{demaine2005subexponential} has led to FPT algorithms with sub-exponential dependence on parameter. Thus a natural question to ask is the following.

\begin{problem}
\label{prob_ecc}
Does \emph{ECC} on planar graphs has an FPT algorithm running in $2^{o(k)}n^{O(1)}$ time?
\end{problem}

From the proof of Lemma \ref{lemma_accs2_nn} (or Theorem \ref{thm_acc3}), it is obvious that on planar graphs \emph{ACC} is solvable in $2^{O(t)}n^{O(1)}$ time. Although, we have not discussed the algorithms for \emph{WECP}, \emph{EWCD}, and \emph{LRVCC} using the framework of \emph{enumerating cliques of restricted subgraphs}, each of these problems are solvable in $2^{O(dk)}n^{O(1)}$ time. Therefore, assuming the problems remain \emph{NP-complete} on planar graphs (complexity of these problems on planar graphs are unresolved), questions similar to Problem \ref{prob_ecc} exist for these problems.

For \emph{ACC}, we have shown an FPT algorithm whose exponent of the running time has quadratic dependency on parameter $t$ (Section \ref{ecrs_acc}). With our new framework, we have shown an FPT algorithm whose exponent of the running time has quasilinear dependency on parameter $t$ (Section \ref{algos_acc}). Thus an important question in this regard is the following.

\begin{problem}
Does \emph{ACC} has an FPT algorithm running in $2^{O(t)}n^{O(1)}$ time?
\end{problem}

\bibliographystyle{plain}

\typeout{}
\bibliography{7_biblio.bib}

\appendix

\section{Proofs omitted in Section \ref{sec_prelim}}
\label{sec:app_prelim}

Using relationship between \emph{degeneracy} and \emph{arboricity}, we can obtain a proof for Lemma \ref{lemma_db} as follows.

\begin{definition}[Arboricity]
\emph{Arboricity} of a graph is the minimum number of spanning forests that cover all the edges of the graph.
\end{definition}

\begin{lemma}
\label{lemma_arbo1}
If $a$ is the arboricity of $G$, then $d \leq 2a -1$.
\end{lemma}
\begin{proof}
Let $H=(V_H, E_H)$ be a subgraph of $G$ with $|V_H| = n_H$ and $|E_H| = m_H$. From Definition \ref{def1}, we have, $$d = \max_{V_H \subseteq V} \min_{x \in V_H} |\{x,y\} \in E_H|.$$

Since the number of spanning forests needed to cover $G$ is $a$, the number of spanning forests needed to cover $H$ is at most $a$, i.e., $m_H \leq a (n_H -1)$. For the average degree of vertices in $H$, we have $\frac{2m_H}{n_H} \leq \frac{2a(n_H -1)}{n_H} < 2a$. Therefore, the minimum degree of vertices in $H$ is at most $2a -1$, i.e., $$\min_{x \in V_H} |\{x,y\} \in E_H| \leq 2a -1.$$

The claim follows, since $H$ is an arbitrary subgraph of $G$.
\end{proof}

Using an equivalent definition of \emph{degeneracy} in terms of acyclic orientation of graph, an exposition of the preceding proof can be found in \cite{chrobak1991planar}.

\begin{lemma}[Lemma 1 \cite{chiba1985arboricity}]
\label{lemma_arbo2}
If $a$ is the arboricity of $G$, then $a \leq \lceil (\sqrt{2m+n})/2 \rceil$.
\end{lemma}

\LemmaDB*
\begin{proof}
Let $G$ has $q$ connected components with $n_1, n_2, \ldots, n_q$ vertices respectively. $m \geq \sum_{i=1}^q n_i - 1 = n -q$. Since $G$ does not contain any isolated vertices $q \leq \lfloor\frac{n}{2} \rfloor$. Therefore, $m \geq n - \lfloor\frac{n}{2} \rfloor \geq \frac{n}{2}$.

Combining Lemma \ref{lemma_arbo1} and Lemma \ref{lemma_arbo2}, we have, 
$$d \leq 2a -1 \leq 2 \lceil (\sqrt{2m+n})/2 \rceil -1 \leq 2 \lceil (\sqrt{2m+2m})/2 \rceil - 1 = 2 \sqrt{m} -1.$$
\end{proof}

\PropEqvOne*
\begin{proof}
\textbf{Only if.} Let $|(x,p)|=q$ and $c>0$ be a constant. For $q \leq f(p)$, we have $f(p)q^{c} \leq \{f(p)\}^{c+1}$. For $q \geq f(p)$, we have $f(p)q^{c} \leq q^{c+1}$. Therefore, $f(p)q^c \leq \max \{\{f(p)\}^{c+1}, q^{c+1}\} \leq \{f(p)\}^{c+1} + q^{c+1} = f^*(p) + |(x,p)|^{O(1)}$.

\textbf{If.} We can assume $f^*(p) \geq 1$ and $|(x,p)|^{O(1)} \geq 1$. We have, $f^*(p) + |(x,p)|^{O(1)} \leq f^*(p)|(x,p)|^{O(1)}$.
\end{proof}

\section{Proofs omitted in Section \ref{sec_algos}}
\label{sec:app_algos}

\LemmaECCSTwoOK*
\begin{proof}
Let $\A$ denote the family of algorithms that consider an arbitrary uncovered edge $\{x,y\}$ at step 2 of \emph{ECCS2}. Clearly, \emph{ECCS2} $\in \A$.

Consider any algorithm $A_i \in \A$. For a fixed value of $k$, consider the family of instances of \emph{ECC}, $\G_k = \{(G,k)\}$. For a search tree $T$ of $A_i$, let us define a node $u^T$ of $T$ to be a \emph{YES} node if $A_i$ returns at step 1 of $u^T$. We claim that for any $(G,k) \in \G_k$, $(G,k)$ is a \emph{YES} instance of \emph{ECC} if and only if $A_i$ would return from a \emph{YES} node. If $A_i$ returns from a \emph{YES} node, then clearly $(G,k)$ is a \emph{YES} instance of \emph{ECC}. It remains to show that for any $(G,k) \in \G_k$, if $(G,k)$ is a \emph{YES} instance of \emph{ECC}, then $A_i$ would return from a \emph{YES} node.

We induct on $k$. For the base case, we have $k=1$, i.e., $G$ is a complete graph. In this case, $A_i$ would create a clique $C_1$ at step 4, only for the very first edge selected at step 2. Any subsequent edges, chosen at step 2 by $A_i$, would only be covered at step 3: this follows from Proposition \ref{prop_isr2}, since $G$ is a complete graph. Therefore, the claim holds for all instances of \emph{ECC} in $\G_1$.

For $k > 1$, assume the claim holds for all instances of \emph{ECC} in $\G_{k^{*}}$ such that $k^{*} < k$. Since $(G,k)$ is a \emph{YES} instance of \emph{ECC}, let $\C = \{C_1, C_2, \ldots, C_k\}$ be a corresponding clique cover. Fix a permutation $\pi$ of $\{1,2, \ldots ,k\}$. From Definition \ref{def_epii}, $E_{\pi(i)} = \{\{x,y\} \in C_{\pi(i)}| \{x,y\} \not\in C_{\pi(j)}, j <i \}$, i.e., the set of edges exclusive to $C_{\pi(i)}$ with respect to the cliques $\{C_{\pi(1)}, C_{\pi(2)}, \ldots, C_{\pi(i-1)}\}$. Let $G' = G \backslash E_{\pi(k)}$. Since $(G,k)$ is a \emph{YES} instance of \emph{ECC}, clearly, $(G^{*}, k-1)$ is a \emph{YES} instance of \emph{ECC}. By inductive hypothesis, for $(G^{*}, k-1)$, $A_i$ would return from a \emph{YES} node. Let $u^T$ be such a \emph{YES} node.

Now, in a search tree $T$ of $A_i$ for $(G,k)$, $A_i$ can cover the edges of $E_{\pi(k)}$ in the descendent nodes of $u^T$. Note that at node $u^T$ in $T$, $A_i$ would use at most $k-1$ cliques. In the descendent nodes of $u^T$ of $T$, the edges of $E_{\pi(k)}$ may be covered at step 3 without creating any new clique. In the worst case, at most one additional clique would be needed to cover the edges of $E_{\pi(k)}$ in the descendent nodes of $u^T$: this follows from Proposition \ref{prop_isr2}, since the end vertices of the edges in $E_{\pi(k)}$ induce a clique in $G$. Since preceding arguments hold for any permutation $\pi$ of $\{1,2, \ldots, k\}$, it follows that for $(G,k)$, $A_i$ would return from a \emph{YES} node. This concludes our inductive step.
\end{proof}

\LemmaECCSTwoSpace*
\begin{proof}
For a search tree $T$ of \emph{ECCS2}, let $u^T$ be a node corresponding to an uncovered edge $\{x,y\}$, selected at step 2. If at node $u^T$ we need to create a new clique to cover the edge $\{x,y\}$, then it must be the case that all branches of step 3 have failed. Any node $v^T$ such that $u^T$ is an ancestor of $v^T$ in the search tree $T$ would not select the edge $\{x,y\}$ as uncovered. Therefore any $x \in V$ would appear at most $|N(x)|$ times in $\C$ and $\sum_{x\in V} |\{l | x \in C_l\}| = O(m)$. Also, using the same argument as in the proof of Proposition \ref{prop_isr3}, for a search tree with at most $k$ cliques, $\sum_{x \in V} |\{l | x \not\in C_l, C_l \subseteq N(x)\}| = O(k\Delta)$. Therefore, $\sum_{x\in V} |R_x| = O(m+k\Delta)$ holds for \emph{ECCS2}.

In aggregate, the sets in $\D$ and the sets in $\R$ hold same information of the relationship among vertices and cliques. To see this, consider the bipartite graph $B=(V,\C,E_B)$, where $E_B= \{\{x,C_l\} | x \in C_l \text{ or } (x \not\in C_l, C_l \subseteq N(x))\}$. In $B$, the sum of the degrees of vertices of the $V$ side must be equal to the sum of the degrees of vertices of the $\C$ side. Therefore, $\sum_{D_l \in \D} |D_l| =\sum_{x\in V} |R_x| = O(m +k \Delta)$.

At each node in the subroutine for step 3a, a vertex of $D_l$ is either kept in $D_l$ or moved to $U$. Therefore, space usage of $U$ data structures in total is already accounted for.
\end{proof}

\LemmaACCSTwoOK*
\begin{proof}
Let $\A$ denote the family of algorithms that consider an arbitrary uncovered edge $\{x,y\}$ at step 2 of \emph{ACCS2}. Clearly, \emph{ACCS2} $\in \A$.

Consider any algorithm $A_i \in \A$. For a fixed value of $t$, consider the family of instances of \emph{ACC}, $\G_t = \{(G,t)\}$. For a search tree $T$ of $A_i$, let us define a node $u^T$ of $T$ to be a \emph{YES} node if $A_i$ returns at step 1 of $u^T$. We claim that for any $(G,t) \in \G_t$, $(G,t)$ is a \emph{YES} instance of \emph{ACC} if and only if $A_i$ would return from a \emph{YES} node. If $A_i$ returns from a \emph{YES} node, then clearly $(G,t)$ is a \emph{YES} instance of \emph{ACC}. It remains to show that for any $(G,t) \in \G_t$, if $(G,t)$ is a \emph{YES} instance of \emph{ACC}, then $A_i$ would return from a \emph{YES} node.

We induct on $t$. For the (non-trivial) base case, we have $t=2$, i.e., $G$ is the graph containing a single edge. In this case, $A_i$ would simply create a clique $C_1$, at step 4 for the only edge of $G$. Therefore, the claim holds for all instances of \emph{ACC} in $\G_t$ such that $t \leq 2$.

For $t > 2$, assume the claim holds for all instances of \emph{ACC} in $\G_{t^{*}}$ such that $t^{*} < t$. Since $(G,t)$ is a \emph{YES} instance of \emph{ACC}, let $\C= \{C_1, C_2, \ldots, C_k\}$ be a corresponding clique cover. If $k=1$, i.e., $G$ is a complete graph, then $(G,t)$ is a \emph{YES} instance of \emph{ACC}. In this case, $A_i$ would create a clique $C_1$ at step 4, only for the very first edge selected at step 2. Any subsequent edges chosen at step 2 by $A_i$ would only be covered at step 3: this follows from Proposition \ref{prop_isr2}, since $G$ is a complete graph. Therefore, for following argument we can assume $k > 1$.

Fix a permutation $\pi$ of $\{1,2, \ldots ,k\}$. From Definition \ref{def_epii}, $E_{\pi(i)} = \{\{x,y\} \in C_{\pi(i)}| \{x,y\} \not\in C_{\pi(j)}, j <i \}$, i.e., the set of edges exclusive to $C_{\pi(i)}$ with respect to the cliques $\{C_{\pi(1)}, C_{\pi(2)}, \ldots, C_{\pi(i-1)} \}$. Let $G' = G \backslash E_{\pi(k)}$ and $t^{*} = t - |V_{\pi(k)}|$ where $V_{\pi(k)} = \{x | \{x,y\} \in E_{\pi(k)}\}$. (Note that if $E_{\pi(k)} = \emptyset$, then, starting from $l=k$, we can keep decreasing $l$ until we find an index $l$ such that $E_{\pi(l)} \ne \emptyset$. Then the following argument would hold verbatim if we replace $E_{\pi(k)}$ with $E_{\pi(l)}$). Since $(G,t)$ is a \emph{YES} instance of \emph{ACC}, $(G^{*}, t^{*})$ is a \emph{YES} instance of \emph{ACC}. By inductive hypothesis, for $(G^{*}, t^{*})$, $A_i$ would return from a \emph{YES} node. Let $u^T$ be such a \emph{YES} node.

Now, in a search tree $T$ of $A_i$ for $(G,t)$, $A_i$ can cover the edges of $E_{\pi(k)}$ in the descendent nodes of $u^T$. Note that at node $u^T$ in $T$, $A_i$ would use at most $t^{*}$ assignments of vertices to the cliques. In the descendent nodes of $u^T$ of $T$, at most $|V_{\pi(k)}|$ additional assignments of vertices to cliques would be needed to cover the edges of $E_{\pi(k)}$: this follows from Proposition \ref{prop_isr2}, since the vertices in $V_{\pi(k)}$ induce a clique in $G$. Since preceding arguments hold for any permutation $\pi$ of $\{1,2, \ldots, k\}$, it follows that for $(G,t)$, $A_i$ would return from a \emph{YES} node. This concludes our inductive step.
\end{proof}

\LemmaAWECPSok*
\begin{proof}
Let $\A$ denote the family of algorithms that consider an arbitrary edge $e = \{x,y\}$ at step 2 of \emph{AWECPS} such that $w^E(e) > 0$. Clearly, \emph{AWECPS} $\in \A$.

Consider any algorithm $A_i \in \A$. For a fixed value of $k$, consider the family of instances of \emph{AWECP}, $\G_k = \{(G,k, w^E, S, w^S)\}$. For a search tree $T$ of $A_i$, let us define a node $u^T$ of $T$ to be a \emph{YES} node if $A_i$ returns at step 1 of $u^T$. We claim that for any $I \in \G_k$, $(G,k)$ is a \emph{YES} instance of \emph{AWECP} if and only if $A_i$ would return from a \emph{YES} node. If $A_i$ returns from a \emph{YES} node, then clearly $I$ is a \emph{YES} instance of \emph{AWECP}. It remains to show that for any $I \in \G_k$, if $I$ is a \emph{YES} instance of \emph{AWECP}, then $A_i$ would return from a \emph{YES} node.

We induct on $k$. For the base case, we have $k=1$, i.e., $G=(V,E)$ is a complete graph, $w^E(e) =1$ for all $e \in E$, and $w^S(x) = 1$ for all $x \in S$. In this case, $A_i$ would create a clique $C_1$ at step 4, only for the very first edge selected at step 2. It is straightforward to see that any subsequent edges chosen at step 2 by $A_i$ would only be included in $C_1$ at step 3. Therefore, the claim holds for all instances of \emph{AWECP} in $\G_1$.

For $k > 1$, assume the claim holds for all instances of \emph{AWECP} in $\G_{k^{*}}$ such that $k^{*} < k$. Since $I = (G,k, w^E, S, w^S)$ is a \emph{YES} instance of \emph{AWECP}, let $\C = \{C_1, C_2, \ldots, C_k\}$ be a corresponding clique cover. We decompose the instance $I = (G,k, w^E, S, w^S)$ into two instances $I_1 = (G^{*},k-1, w^{*}, S^{*}, w^{S^{*}})$ and $I_2 = (H,1,w^{H}, S_H, w^{S_H})$ such that $I_1 \in $ \emph{AWECP} and $I_2 \in $ \emph{AWECP}.

\textbf{Construction of $I_2$}: Let $H=(V_H, E_H)$ be the subgraph induced by $C_l \in \C$ in $G$, i.e., $H=G[C_l]$. Let $w^{H}(e) = 1$ for all $e \in E_H$. Let $S_H = S \cap V_H$, and $w^{S_H}(x) = 1$ for all $x \in S_H$.

\textbf{Construction of $I_1$}: Let $G^{*} = (V^{*}, E^{*})$ where $E^{*} = E \backslash \{e \in E_H | w^E(e) = 1\}$ and $V^{*} = \{x | \{x,y\} \in E^{*}\}$. For all $e \in E^{*} \cap E_H$, let $w^{*}(e) = w^E(e) -1$, and for all $e \in E^{*} \backslash E_H$, let $w^{*}(e) = w^E(e)$. Let $S^{*} = S \backslash \{x \in S \cap V_H| w^S(x) = 1\}$. For all $x \in S^{*} \cap V_H$, let $w^{S^{*}}(x) = w^S(x) -1$, and for all $s \in S^{*} \backslash V_H$, let $w^{S^{*}}(x) = w^S(x)$.

It is straightforward to verify that the composition (union of the graphs and addition of the parameters following the decomposition) of $I_1$ and $I_2$ is exactly the instance $I = (G,k,w^E,S,w^S)$. Since $I$ is a \emph{YES} instance of \emph{AWECP}, both $I_1$ and $I_2$ are \emph{YES} instances of \emph{AWECP}. By inductive hypothesis, for both $I_1$ and $I_2$, $A_i$ would return from a \emph{YES} node. Let $u^T$ be a \emph{YES} node for $I_1$.

Now, in a search tree $T$ of $A_i$ for $I$, $A_i$ can construct the solution for $I_2$ in the descendent nodes of $u^T$, and append the solution of $I_2$ to the solution of $I_1$. Note that at node $u^T$ in $T$, $A_i$ would use at most $k-1$ cliques. In the descendent nodes of $u^T$ of $T$, the solution for $I_2$ may be exclusively constructed at step 3 without creating any new clique. In the worst case, at most one additional clique would be needed to construct a solution for $I_2 = (H,1,w^H,S_H,w^{S_H}$), since $I_2$ is a \emph{YES} instance of \emph{AWECP}.

Note that the permutation of cliques in $\C$ does not affect the preceding arguments, and the arguments hold for any clique $C_l \in \C$. Therefore, it follows that for the instance $(G,k, w^E, S, w^S)$, $A_i$ would return from a \emph{YES} node. This concludes our inductive step.
\end{proof}

\end{document}